\documentclass[11pt,reqno]{amsart}

\usepackage{amsmath}
\usepackage{amssymb}
\usepackage{a4wide}
\usepackage{bbm}
\usepackage{graphics}
\usepackage{color}
\usepackage[sans]{dsfont}

\usepackage{enumerate}

     \newcommand{\Ran}{{\operatorname{Ran}}}
\newcommand{\Ker}{{\operatorname{Ker}}}

     \newcommand{\R}{{\mathbb{R}}}

\newcommand{\w}{{\rm w}}

\newcommand{\e}{{\rm e}}
\renewcommand{\i}{{\rm i}}
\renewcommand{\d}{{\rm d}}

\renewcommand{\Re}{{\rm Re}\,}
\renewcommand{\Im}{{\rm Im}\,}

\newcommand\inp[2][]{#1 \langle #2#1\rangle}
\newcommand\parb[2][]{#1 \big ( #2#1\big )}

\newcommand{\pp}{{\rm pp}}

\newcommand{\mand}{\text{ and }}
\newcommand{\mfor}{\text{ for }}
\newcommand{\D}{\mathcal{ D }}

     \theoremstyle{plain}
     \newtheorem{thm}{Theorem}[section]
     \newtheorem{prop}[thm]{Proposition}
     \newtheorem{lemma}[thm]{Lemma}
      \newtheorem{cor}[thm]{Corollary}
     \theoremstyle{definition}
     
     \newtheorem{defn}[thm]{Definition}

     \newtheorem{cond}[thm]{Condition}
\newtheorem{conds}[thm]{Conditions}

     \newtheorem{remark}[thm]{Remark}
     \newtheorem{remarks}[thm]{Remarks}
\newtheorem*{remarks*}{Remarks}
\newtheorem*{remark*}{Remark}
\newtheorem*{cond*}{Condition}

     \numberwithin{equation}{section}

\title[Second order perturbation theory]{Second order perturbation theory for embedded eigenvalues}

\author{J. Faupin}
\address[J. Faupin] {Institut de Math{\'e}matiques de Bordeaux, UMR-CNRS 5251 \newline Universit{\'e} de Bordeaux 1, 351 cours de la lib{\'e}ration, 33405 Talence Cedex, France \newline \textit{Partially supported by the Center for Theory in Natural Sciences, Aarhus University}}
 \email{jeremy.faupin@math.u-bordeaux1.fr} 

\author{J.S. M\o ller}
\address[J.S.  M\o ller]{Institut for  Matematiske
Fag \newline
Aarhus Universitet\\ Ny Munkegade, 8000 Aarhus C, 
Denmark}
\email{jacob@imf.au.dk}

\author{E. Skibsted}
\address[E. Skibsted]{Institut for  Matematiske
Fag \newline
Aarhus Universitet\\ Ny Munkegade, 8000 Aarhus C, 
Denmark}
\email{skibsted@imf.au.dk}

\date{\today}

\begin{document}

\begin{abstract}
We study second order perturbation theory for embedded eigenvalues of
an abstract class of self-adjoint operators. Using an extension of the
Mourre theory, under assumptions on the regularity of bound states
with respect to a conjugate operator, we prove upper semicontinuity of
the point spectrum and establish the Fermi Golden Rule criterion. Our
results apply to massless Pauli-Fierz Hamiltonians for arbitrary
coupling.
\end{abstract}

\maketitle
\tableofcontents

\section{Introduction}\label{Introduction}

In this second of a series of papers, we study second order
perturbation theory for embedded eigenvalues of an abstract class of
self-adjoint operators. Perturbation theory for isolated eigenvalues
of finite multiplicity is well-understood, at least if the family of
operators under consideration is analytic in the sense of Kato (see
\cite{Ka,RS}). The question is more subtle when dealing with
unperturbed eigenvalues embedded in the continuous spectrum. A method
to tackle this problem, which we shall not develop here, is based on
analytic deformation techniques and gives rise to a notion of
resonances. It appeared in \cite{AC,BC} and was further extended by
many authors in different contexts (let us mention \cite{Si,RS, JP,BFS}
among many other contributions).  As shown in \cite{AHS}, another way
of studying the behaviour of embedded eigenvalues under perturbation
is based on Mourre's commutator method (\cite{Mo}). We shall develop
this second approach from  an abstract point of view in this paper.

We mainly require two conditions: The first one corresponds to a set of assumptions needed in order to use the Mourre method (see Conditions \ref{cond:perturbation1} below). We shall work with an extension of the Mourre theory which we call \emph{singular Mourre theory}, and which is closely related to the ones developed in \cite{Sk,MS,GGM1}. Singular Mourre theory refers to the situation where the commutator of the Hamiltonian with the chosen ``conjugate operator'' is not controlled by the Hamiltonian itself. The \emph{regular Mourre theory}, studied for instance in \cite{Mo,ABG,HuSp,HS,Ca,CGH}, is a particular case of the theory considered here. A feature of singular Mourre theory is to allow one to derive spectral properties of so-called Pauli-Fierz Hamiltonians. This shall be discussed in Section \ref{section:Pauli-Fierz}.

Our second set of assumptions concerns the regularity of bound states with respect to a conjugate operator (see Conditions \ref{cond:perturbation2b}, \ref{cond:perturbation2c} and \ref{cond:perturbation2} below). Related questions are discussed in details, in an abstract framework, in the companion paper \cite{FMS} (see also \cite{Ca,CGH}).

Our main concerns are to study upper semicontinuity of point spectrum
(Theorem \ref{thm:upper}) and to show that the Fermi Golden Rule
criterion  (Theorem \ref{thm:fermi-golden-rule}) holds. If the Fermi
Golden Rule condition is not fulfilled  we shall still obtain an expansion to second order of perturbed eigenvalues. Before precisely stating our results and comparing them to the literature, we introduce the abstract framework in which we shall work.

\subsection{Assumptions}\label{subsection:assumptions}

Let $\mathcal{H}$ be a complex Hilbert space. Suppose that $H$ and $M$
are self-adjoint operators on $\mathcal{H}$, with $M \ge 0$, and
suppose that a symmetric operator $R$ is given such that $\D(R)
\supseteq \D(H)$. Let $H' := M+R$ defined on
$\D:=\D(M) \cap \D(H)$. Let $\mathcal{G} := \D( |H|^{1/2} ) \cap \D(
M^{1/2} )$ equipped with the intersection topology. Let $A$ be a
closed, densely defined, maximal symmetric operator on $\mathcal{H}$. In particular, introducing
deficiency indices $n_{\mp}=\dim \Ker (A^*\pm \i)$, either $n_+=0$ or
$n_-=0$. If $n_+=0$ (or $n_-=0$) $B:=A$ (or $B:=-A$) generates a $C_0$-semigroup of isometries $\{ W_t \}_{t \ge 0}$ (see \cite{RS,ABG}). We recall that a map $\mathbb{R}^+ \ni t \mapsto W_t \in \mathcal{B}( \mathcal{ H } )$ is called a $C_0$-semigroup if $W_0=I$, $W_tW_s=W_{t+s}$ for $t,s\ge0$, and $w\text{-}\lim_{t\to 0^+} W_t = I$. Here $\mathcal{B}( \mathcal{H} )$ denotes the set of bounded operators on $\mathcal{H}$. The generator $B$ of a $C_0$-semigroup $\{ W_t \}_{t \ge 0}$ is defined by
\begin{equation}
\D (B) := \{ u \in \mathcal{H} , Bu := \lim_{t \to 0^+} ( \i t )^{-1} ( W_t u - u ) \text{ exists} \}.
\end{equation} 
We write $W_t = e^{\i tB}$. 

For any Hilbert spaces $\mathcal{H}_1$ and $\mathcal{H}_2$, we denote by $\mathcal{B}( \mathcal{H}_1 ; \mathcal{H}_2 )$ the set of bounded operators from $\mathcal{H}_1$ to $\mathcal{H}_2$. We use the notation $\langle B \rangle := ( 1 + B^*B )^{1/2}$ for any closed operator $B$. Throughout the paper, $C_j$, $j=1,2,\dots$, will denote positive constants that may differ from one proof to another. Let us recall the following definition from \cite{GGM1}:
\begin{defn}
Let $\{ W_{1,t} \}$, $\{W_{2,t} \}$ be two $C_0$-semigroups on Hilbert spaces $\mathcal{H}_1$, $\mathcal{H}_2$ with generators $A_1$, $A_2$ respectively. A bounded operator $B \in \mathcal{B}( \mathcal{H}_1 ; \mathcal{H}_2) $ is said to be in $\mathrm{C}^1(A_1 ; A_2)$ if
\begin{equation}
\| W_{2,t} B - B W_{1,t} \|_{\mathcal{B}( \mathcal{H}_1 ; \mathcal{H}_2 )} \le C t , \quad 0 \le t \le 1,
\end{equation}
for some positive constant $C$. 
\end{defn}
\begin{remarks}\begin{enumerate}[\quad\normalfont 1)] 
\item \label{item:f1} By \cite[Proposition 2.29]{GGM1}, $B\in \mathcal{B}( \mathcal{H}_1 ;
\mathcal{H}_2) $ is of class $\mathrm{C}^1( A_1 ; A_2 )$ if and only
if the sesquilinear form $_2[B,\i A]_1$ defined on $\D(A_2^*) \times
\D(A_1)$ by $\inp {\phi,  {_2}[B,\i A]_1 \psi } = \i\inp {B^* \phi ,
  A_1 \psi } - \i\inp {A_2^* \phi , B \psi}$ is bounded relatively to
the topology of $\mathcal{H}_2 \times \mathcal{H}_1$. The associated
bounded operator in $\mathcal{B}( \mathcal{H}_1 ; \mathcal{H}_2 )$ is
denoted by $[B,\i A]^0$ and we have
\begin{equation}\label{eq:vir}
[B,\i A]^0 = \mathrm{s}\text{-}\lim_{t\to0^+}  t^{-1} [ B W_{1,t} - W_{2,t} B ].
\end{equation}
\item \label{item:f2}We  recall (see \cite{ABG}) that if $A$ and $B$
  are  self-adjoint operators on a Hilbert space $\mathcal 
H$, $B$ is said to be in $\mathrm{C}^1(A)$ if there exists $z\in
\mathbb{C} \setminus \mathbb{R}$ such that $(B-z)^{-1}\in
\mathrm{C}^1( A ; A )$ (meaning here that $\mathcal 
H_j=\mathcal 
H$  and $A_j=A$, $j=1,2$). In that case in fact $(B-z)^{-1}\in
\mathrm{C}^1( A ; A )$ for all $z\in \rho(B)$ ($\rho(B)$ is the resolvent set of $B$).
\item \label{item:f3} The standard Mourre class, cf. \cite {Mo}, is a
  subset of $\mathrm{C}^1(A)$ given as follows: Notice that for any
  $B\in \mathrm{C}^1(A)$ the
  commutator form $[B,\i A]$ defined on $\D(B)\cap \D(A)$ extends
  uniquely (by continuity) to a bounded form $[B,\i A]^0 $ on
  $\D(B)$. We shall say that $B$ is {\it Mourre-$C^1(A)$} if $[B,\i A]^0
  $ is a $B$-bounded operator on $\mathcal 
H$. The  subclass of {\it Mourre-$C^1(A)$} operators in $\mathrm{C}^1(A)$ is in this paper denoted by $\mathrm{C}^1_{\rm Mo}(A)$.
\end{enumerate}
\end{remarks}

Let us now state our first set of conditions:

\begin{conds}\label{cond:perturbation1}$\quad$

 \begin{enumerate}[\quad\normalfont (1)]
\item \label{item:g1} $H \in \mathrm{C}^1_{\rm Mo}(M)$. 
\item \label{item:g2} There is an interval $I \subseteq \mathbb{R}$ such that for all $\eta \in I$, there exist $c_0>0$, $C_1 \in \mathbb{R}$, $f_\eta \in \mathrm{C}_0^\infty( \mathbb{R})$, $0\le f_\eta \le 1$ and $f_\eta=1$ in a neighbourhood of $\eta$, and a compact operator $K_0$ on $\mathcal{H}$ such that, in the sense of quadratic forms on $\mathcal{D}$,
\begin{equation}\label{eq:Mourre_estimate}
M + R \ge \mathrm{c}_0 I - C_1 f_\eta^\perp(H)^2 \langle H \rangle - K_0,
\end{equation}
where $f_\eta^\perp(H) = 1 - f_\eta(H)$.
\item \label{item:g3}$\mathcal{G}$ is ``\emph{boundedly-stable}'' under $\{ W_t \}$ and $\{ W_t^* \}$ i.e. $W_t \mathcal{G} \subseteq \mathcal{G}$, $W_t^* \mathcal{G} \subseteq \mathcal{G}$, $t  > 0$, and for all $\phi \in \mathcal{G}$,
\begin{equation}
\sup_{0<t<1} \| W_t \phi \|_\mathcal{G} <  \infty, \quad \sup_{0<t<1} \| W_t^* \phi \|_\mathcal{G} < \infty.
\end{equation}
Let $A_\mathcal{G}$
denote the generator of the $C_0$-semigroup $W_t |_\mathcal{G}$ 
(see \linebreak \cite[Lemma 2.33]{GGM1} for justification). Let $A_{\mathcal{G}^*}$ denote the
generator of the $C_0$-semigroup given as  the extension of $W_t$ to $\mathcal{G}^*$. 
\item \label{item:g4}$H\in \mathrm{C}^1(
  A_\mathcal{G};A_{\mathcal{G}^*})$ and the operator $H' = M+R$ satisfies $H' = [H,\i A]^0 \in
  \mathrm{C}^1( A_\mathcal{G};A_{\mathcal{G}^*})$ (see Remark
  \ref{rk:perturbation1} \ref{item:h1}) for justification of notation). We set $H'':=[H',\i A]^0$.
\end{enumerate}
\end{conds}
\begin{remarks}\label{rk:perturbation1} \begin{enumerate}[\quad\normalfont 1)]
\item \label{item:h1} It follows from Condition
  \ref{cond:perturbation1} (\ref{item:g1}) that $\D
  $ is a core for $H$ as well as for $M$,
  cf. \cite{ABG,GG}. This condition is transcribed from \cite{Sk} and is
  stronger than \cite[(M1)]{GGM1}, cf. \cite[Lemma 2.26]{GGM1}. Another
  consequence of Condition
  \ref{cond:perturbation1} (\ref{item:g1}) is the following
  alternative description of the space $\mathcal G$: Let $G$ be the
  Friedrichs extension of the operator $M+\inp{H}$ on
  $\mathcal{D}$. Then $\D(\sqrt G)=\mathcal G$; this follows from
  \cite[Proposition 3.8]{GGM1}. (For the readers convenience  we remark
  that the statement   actually can be
  proved directly by using elementary interpolation,
  cf. \cite[(3.14)]{FMS}.) In particular $\D
  $ is dense in $\mathcal G$,  and we can consider
  $H'$ as a bounded operator in the sense used in  Condition
  \ref{cond:perturbation1} (\ref{item:g4}): $H'\in \mathcal B(\mathcal
  G; \mathcal G^*)$. 
\item \label{item:h2} Suppose  Conditions
  \ref{cond:perturbation1}. Then the following identity holds for all
  $\phi_1 \in \D \cap \D(A^*)$ and  $\phi_2 \in \D \cap \D(A)$:
\begin{equation}\label{eq:5a}
\langle \phi_1 , (M + R ) \phi_2 \rangle = \i \langle H \phi_1 , A \phi_2 \rangle - \i \langle A^* \phi_1 , H \phi_2 \rangle.
\end{equation} This is a consequence of (\ref{eq:vir}). Another
(related) consequence of (\ref{eq:vir}) is the following  version of the
so-called virial theorem: For any eigenstate, $(H-\lambda)\psi=0$, with  $\psi\in \D (M^{1/2})$
\begin{equation}
  \label{eq:5b}
  \langle \psi , (M + R ) \psi \rangle = 0.
\end{equation}
\item \label{item:h3} The conditions of the regular Mourre theory
  considered for instance in \cite{Mo,ABG,HuSp,HS,Ca,CGH} constitute a particular
  case of Conditions \ref{cond:perturbation1} assuming that $M=0$. In
  \cite{Mo,ABG,HS,Ca,CGH}, the conjugate operator $A$ is supposed to
  be self-adjoint, whereas in \cite{HuSp} the weaker assumption that
  $A$ is the generator of a $C_0$-semigroup of isometries is required.
  Notice that in the case where $M=0$ and $A$ is
  self-adjoint  Condition
  \ref{cond:perturbation1} (\ref{item:g3}) appears  replaced by the stronger
  condition: 
$\sup_{|t|<1} \|\e^{\i t A} \phi\|_{\D (H)} <  \infty$ for any
$\phi\in \D (H)$. By \cite[Lemma 10.2.1]{HP} if $\D(H)$ is a separable Hilbert space
 the latter condition 
 is a consequence of the fact that $e^{ \i t A } \D( H ) \subseteq \D( H )$ for all $t \in \mathbb{R}$.
 It should also be noticed that the
  boundedness of $H''$ with respect to $H$ is often required in the
  regular Mourre theory. Condition \ref{cond:perturbation1} (\ref{item:g4}) leads
  to the weaker assumption that $\langle H \rangle^{-1/2} H'' \langle H
  \rangle^{-1/2}$ is bounded.
  \item \label{item:h5} The idea of splitting the formal commutator $\i
  [H,A]$ into an $H$-unbounded piece, $M$, and a $H$-bounded piece,
  $R$, appeared first in  \cite{Sk}. As it was shown in \cite{Sk}, and
  later in \cite{GGM}, this extension of the Mourre theory allows one
  to study spectral properties of $N$-body systems coupled to bosonic
  fields (see also \cite{MS} for the use of related assumptions in a
  different context). This will be discussed more precisely in the
  next section. 
\item \label{item:h4} We notice that Conditions
  \ref{cond:perturbation1} (with $K_0=0$ in (\ref{item:g2})) are stronger than Hypotheses (M1)--(M5) used in \cite{GGM1} (the operator $H'$ in \cite{GGM1} is supposed to be closed; it corresponds to the closure of the operator $H'$ considered in this paper). Therefore, in particular, the results proved in \cite{GGM1} hold under Conditions \nolinebreak \ref{cond:perturbation1}.
\end{enumerate}
\end{remarks}

Throughout the discussion below we impose (mostly tacitly) Conditions
  \ref{cond:perturbation1}.
We introduce the following  classes of operators (to be considered as
classes of ``perturbations''): 
\begin{defn}
  \label{defn:results-under-least} We say that a symmetric operator
  $V$ with $\mathcal{D}(V)\supseteq\mathcal{D}(H)$,  $\epsilon$-bounded  relatively to $H$, is in $\mathcal V_1$ if $V\in \mathrm{C}^1(
  A_\mathcal{G};A_{\mathcal{G}^*})$ and $V' := [ V , \i A ]^0$ is given as an $H$-bounded
  operator. 
For any $V \in \mathcal V_1$, we set 
\begin{align}
  &\| V \|_1 := \| V (H-\i)^{-1} \|+
  \| V' (H-\i)^{-1}\|.
\end{align}
\end{defn}

It follows from the Kato-Rellich Theorem that for any  $V\in \mathcal
V_1$ the operator $H+V$ is self-adjoint with $\D(H+V) = \D(H)$. 

\begin{defn}\label{def:W2}
We say that $V\in \mathcal{V}_1$  is in $\mathcal V_2$ if $V'\in \mathrm{C}^1(
  A_\mathcal{G};A_{\mathcal{G}^*})$,
  and we set
\begin{align}
  \label{eq:5bis}
  &\| V \|_2 := \| V \|_1+ \| V'' \|_{\mathcal B (\mathcal G; \mathcal G^*)},
\end{align} 
where $V'' := [ V' , \i A ]^0$.
\end{defn}

Our main assumptions on the unperturbed eigenstates are stated in Condition \ref{cond:perturbation2b} and in its stronger version Condition \ref{cond:perturbation2c}.

\begin{cond}\label{cond:perturbation2b}
If $\lambda \in I$ is an eigenvalue of $H$, any eigenstate $\psi$ associated to $\lambda$, $H \psi = \lambda \psi$, satisfies $\psi \in \D(A) \cap \D(M)$. 
\end{cond}
\begin{remark} \label{remarks:domain} Under Condition
  \ref{cond:perturbation2b} and with $\psi$ given as there, one
  verifies using \eqref{eq:vir} and  the fact that
  $\D$ is dense in $\D(H)$  that  $\psi \in \D (HA):=\{\phi\in
  \D (A)|\,A\phi\in \D (H)\}$, cf. Remark
  \ref{rk:perturbation1} \ref{item:h2}).
\end{remark}
\begin{cond}\label{cond:perturbation2c}
If $\lambda \in I$ is an eigenvalue of $H$, any eigenstate $\psi$ associated to $\lambda$, $H \psi = \lambda \psi$, satisfies $\psi \in \D(A^2) \cap \D(M)$. 
\end{cond}

The (possibly existing) perturbed eigenstates may fulfil the following condition:

\begin{cond}\label{cond:perturbation2} For any compact
  interval $J \subseteq I$ there exist 
  $\gamma>0$ and a subset $\mathcal{B}_{1,\gamma}$ of the ball centered at 0 with radius $\gamma$ in  $\mathcal{V}_1$,
  \begin{equation}
  \mathcal{B}_{1,\gamma} \subseteq \{ V \in \mathcal V_1 , \| V \|_1 \le \gamma \},
  \end{equation}
  such that $0 \in \mathcal{B}_{1,\gamma}$, $\mathcal{B}_{1,\gamma}$ is star-shaped and symmetric with respect to 0, and the following holds: There exists $C>0$
  such that, if $V\in \mathcal{B}_{1,\gamma}$ and $(H+V-\lambda) \psi=0$
  with $\lambda\in J$, then $\psi \in \D(A) \cap \D(M)$ and
  \begin{equation}
    \label{eq:12i}
    \|A\psi\| \leq C \|\psi\|.
  \end{equation}
\end{cond}

The following  two conditions are needed for our version of
the so-called Fermi Golden Rule criterion. The first condition is a
technical addition to Conditions
  \ref{cond:perturbation1}:

\begin{cond}\label{cond:tech} 
$\D (M^{1/2})\cap \D (H)\cap \D(A^*)$ is dense in $\D(A^*)$.
\end{cond}
\begin{remarks}\label{remarks:ff} 
\begin{enumerate}[\quad\normalfont 1)]
\item \label{item:p1} Suppose the following modification of the part
  of Condition \ref{cond:perturbation1} (\ref{item:g3}) concerning
  the adjoint semigroup: $\D$ is bounde\-dly-stable under $\{ W_t^* \}$ i.e. $W_t^* \mathcal{D} \subseteq \mathcal{D}$, $t  > 0$, and for all $\phi \in \mathcal{D}$,
\begin{equation}\label{eq:5D}
\sup_{0<t<1} \| W_t^* \phi \|_\mathcal{D} < \infty.
\end{equation} Then $\D \cap \D(A^*)$ is dense in $\D(A^*)$,
cf. \cite[Remark 2.35]{GGM1}. This statement is of course stronger than Condition
  \ref{cond:tech}.
\item \label{item:p2} In our applications  Condition
  \ref{cond:tech}  can be avoided upon changing the definition of
  $\mathcal{V}_1$. Explicitly this modification is given by imposing  in Definition \ref{defn:results-under-least}
  the following additional  condition (replacing $\epsilon$-boundedness with respect to $H$): $V$ is
  $\inp{H}^{1/2}$-bounded. (See  Remark 
\ref{rk:Fermi} \ref{item:fermi1}).)
\end{enumerate}
\end{remarks}

Our second  condition is the so-called Fermi Golden Rule condition. 

\begin{cond}
  \label{cond:fermi-golden-rule} Suppose Conditions
  \ref{cond:perturbation2b} and 
  \ref{cond:tech}. Suppose $\lambda\in \sigma_\pp( H)$ and
  let $P$ denote the eigenprojection $P=E_H(\{\lambda\})$ and $\bar P=I-P$. For given
  $V \in \mathcal V_1$  there exists  $c>0$ such that 
  \begin{equation}
    \label{eq:18}
    P V \Im \parb{ (H-\lambda-\i 0)^{-1}\bar P} V P\geq c P.
  \end{equation}
\end{cond}

We shall see in Section \ref{Reduced limiting absorption principle at
  an eigenvalue} that the left-hand-side of \eqref{eq:18} defines a
bounded operator for any $V \in \mathcal V_1$ (see  Remark 
\ref{rk:Fermi} \ref{item:fermi1}) for  details). This point might be  surprising for
the reader due to the low
degree of regularity imposed  by Condition
  \ref{cond:perturbation2b} (for example $P$ may not map into
  $\D(A^2)$ under the stated conditions, see the end of the next
  subsection for a further discussion).

\subsection{Main results}\label{subsection:main_results}

We have the following result on  upper semicontinuity of the point spectrum of $H$, showing, in other words, that the total multiplicity of the perturbed eigenvalues near an unperturbed one, $\lambda$, cannot exceed the multiplicity of $\lambda$.
\begin{thm}
  \label{thm:upper} Assume that Conditions \ref{cond:perturbation1}
  and Condition \ref{cond:perturbation2} hold. Let $\lambda\in I$ and $J\subseteq I$ be a compact interval
  including $\lambda$ such that $\sigma_\pp( H)\cap J = \{\lambda\}$. Fix $\gamma > 0$ and $\mathcal B_{1,\gamma}$ as in Condition \ref{cond:perturbation2}. There exists $0 < \gamma' \le \gamma$ 
  such that if $V\in \mathcal{B}_{1,\gamma}$ and $\| V \|_1\leq \gamma'$, the total multiplicity of the
  eigenvalues of $H+V$ in $J$ is at most $\dim \Ker (H-\lambda)$.
\end{thm}
Notice that the appearing quantity  $\dim \Ker (H-\lambda)$ is
finite. This is in fact a
consequence of Conditions \ref{cond:perturbation1} and Condition
\ref{cond:perturbation2b}, cf. Remark \ref{rk:perturbation1}
\ref{item:h2}).
 We remark that Theorem \ref{thm:upper} is an abstract version of
\cite[Theorem 2.5]{AHS} where  upper semicontinuity of the point
spectrum of $N$-body Schr{\"o}dinger operators is established. The
proof is essentially the same.

In the case where $H$ does not have eigenvalues in $J$, we do not need Condition \ref{cond:perturbation2} to establish upper semicontinuity of point spectrum. More precisely, we will prove that $\sigma_\pp( H + \sigma V) \cap J = \emptyset$ for $|\sigma|$ small enough under the condition that $V \in \mathcal{V}_2$ (see Corollary \ref{cor:stability}). If it is only required that $V \in \mathcal{V}_1$, the result still holds true provided we assume in addition that any eigenstate of $H+\sigma V$ belongs to $\D( M^{ 1/2 } )$ (see Corollary \ref{cor:stability2}).

One might suspect that there is a similar semistability result as the one stated in Theorem \ref{thm:upper} given upon replacing Condition \ref{cond:perturbation2} by Condition \ref{cond:perturbation2c} (assuming now smallness of $\|V\|_2$). Although there is a formal argument, Conditions \ref{cond:perturbation1} are insufficient for a rigorous proof. Nevertheless the analogous assertion is true in the special case where $H$ does not have eigenvalues in the interval $J$, cf.  Corollary \ref{cor:stability}. Notice also that another special case, although treated under additional conditions, is part of Theorem \ref{thm:fermi-golden-rule} stated below.

For any  $V \in \mathcal V_1$ and $\sigma\in \R$ we set $H_\sigma := H + \sigma V$. A 
main result of this paper is the following assertion on absence of
eigenvalues of $H_\sigma$ for small non-vanishing $|\sigma|$ and  for 
a $V$ fulfilling (\ref{eq:18}):
\begin{thm}
  \label{thm:fermi-golden-rule} Assume that Conditions
  \ref{cond:perturbation1}, Condition \ref{cond:perturbation2b}  and   Condition
  \ref{cond:tech} hold. Assume that   Condition
  \ref{cond:fermi-golden-rule} holds for some $V \in \mathcal V_1$. Let $J \subseteq I$ be any  compact interval such that $\sigma_\pp( H)\cap J= \{\lambda\}$.  Suppose one of the
  following two conditions:
\begin{enumerate}[\quad\normalfont i)]
\item \label{item:a} Condition \ref{cond:perturbation2c} and $V \in \mathcal V_2$.
\item \label{item:b} Condition \ref{cond:perturbation2} and $V \in \mathcal B_{1,\gamma}$.
\end{enumerate}
There exists $\sigma_0>0$ such that for all $\sigma \in ]-\sigma_0,\sigma_0[\;\setminus\{0\}$,
  \begin{equation}
    \label{eq:21}
  \sigma_\pp( H_\sigma)\cap J= \emptyset. 
  \end{equation}
\end{thm}

This type of theorem is usually referred to as the  Fermi Golden Rule
criterion (or in short just Fermi Golden Rule).
In the framework of regular Mourre theory (that is in particular if $M=0$, see Remark
\ref{rk:perturbation1} \ref{item:h3}) above), if $A$ is self-adjoint, Fermi Golden Rule  is well-known. It was first proved in \cite{AHS} for $N$-body Schr{\"o}dinger operators, under an assumption of the type $V \in \mathcal V_2$ and using exponential bounds for eigenstates (yielding in particular an analogue of Condition \ref{cond:perturbation2c}). In \cite{HS}, Theorem \ref{thm:fermi-golden-rule} is proved in an abstract setting assuming Condition \ref{cond:perturbation2c} and the  $H$-boundedness of $V''$. In \cite{Ca,CGH}, still in the framework of regular Mourre theory and with $A$ self-adjoint, it is shown that an assumption of the type $H \in \mathrm{C}^4(A)$ implies Condition \ref{cond:perturbation2c}. A similar result also appears in \cite{GJ} under slightly weaker (``local'') assumptions, still requiring, however, the boundedness of four commutators.

Theorem \ref{thm:fermi-golden-rule} improves the previous results for
the following two reasons: First, as mentioned above, Conditions
\ref{cond:perturbation1} do not require that $A$ be self-adjoint
neither that the formal commutator $\i [H,A]$ be $H$-bounded, which
can be important in applications (see in particular Section
\ref{section:Pauli-Fierz} on Pauli-Fierz Hamiltonians). Second, we
prove that the Fermi Golden Rule criterion also holds under Condition
\ref{cond:perturbation2} and the hypothesis $V \in \mathcal B_{1,\gamma}$ (that
is under condition \ref{item:b}) of Theorem
\ref{thm:fermi-golden-rule}), which to our knowledge  constitutes a
new result even in the framework of regular Mourre theory. Let us
emphasize that Condition \ref{cond:perturbation2} does not contain the
assumption that the eigenstates are in the domain of $A^2$, but only
in the domain of $A$. The price we have to pay lies in the fact that
Condition \ref{cond:perturbation2} involves information on the
possibly existing perturbed eigenstates, which  in concrete models might (at a first
glance) seem rather difficult to obtain.

Nevertheless in a separate paper, \cite{FMS} we provide
abstract hypotheses
under which Condition \ref{cond:perturbation2} is indeed satisfied. As a
consequence, we obtain that Theorem \ref{thm:fermi-golden-rule} applies
for a class of Quantum Field Theory models provided that the Hamiltonian only
has two bounded commutators with $A$ (defined in a suitable sense), see Section
\ref{section:Pauli-Fierz}. We emphasize that from an abstract point of
view, working with $C^2(A)$ conditions, in fact verifying Condition
\ref{cond:perturbation2} is {\it doable} while Condition
\ref{cond:perturbation2c} might be {\it false}, see \cite[Example 1.4]{FMS} for a counterexample.

Recently Rasmussen together with one us (\cite{MR}) studied the essential energy-momentum spectrum of the translation invariant massive Nelson Hamiltonian $H$. In particular the authors construct, for a given total momentum $P$
and non-threshold energy $E$, a conjugate operator $A$
with respect to which the fiber Hamiltonian $H(P)$ satisfies a Mourre estimate, locally uniformly in $E$ and $P$.
From the point of view of the present paper this model is of interest because $H(P)$ is of class $\mathrm{C}^2(A)$ but (presumably) not of class $\mathrm{C}^3(A)$.
This means that, even though the context of \cite{MR} is regular Mourre theory, the improvements of this paper and its companion \cite{FMS} are both essential
to conclude anything about the structure of embedded non-threshold eigenvalue bands.

We shall use different methods to prove Theorem
\ref{thm:fermi-golden-rule}  depending on whether we assume
\ref{item:a}) or \ref{item:b}). In the first case, we shall obtain
an expansion to second order of any possibly existing perturbed
eigenvalue near the unperturbed one $\lambda$. In the second case, \ref{item:b}), this will also be done under the further hypothesis $\dim \mathrm{Ran}(P) = 1$, but we shall proceed differently if the unperturbed eigenvalue is degenerate. In both cases, a key ingredient of the proof consists in obtaining a ``reduced Limiting Absorption Principle'' at an eigenvalue (see Theorems \ref{thm:reduc-limit-absorpt2} and \ref{thm:reduc-limit-absorpt3} below).

The paper is organized as follows: In the next section, we consider Pauli-Fierz Hamiltonians which constitute our main example of a model satisfying the abstract conditions stated above. Section \ref{Reduced limiting absorption principle at an eigenvalue} concerns reduced Limiting Absorption Principles at an eigenvalue $\lambda$ of $H$. In Section \ref{Upper semicontinuity of point spectrum}, we study upper semicontinuity of point spectrum and prove Theorem \ref{thm:upper}. Finally in Section \ref{Second order perturbation theory}, we study second order perturbation theory assuming either Condition \ref{cond:perturbation2c} or Condition \ref{cond:perturbation2}, and we prove Theorem \ref{thm:fermi-golden-rule}.


\section{Application to the spectral theory of Pauli-Fierz models}\label{section:Pauli-Fierz}

\subsection{Massless Pauli-Fierz Hamiltonians} \label{subsection:def-Pauli-Fierz}

The main example we have in mind fitting into the framework of Section \ref{Introduction} consists of an abstract class of Quantum Field Theory models, sometimes called massless Pauli-Fierz models (see for instance \cite{DG,DJ,GGM,FMS}). The latter describe a ``small'' quantum system linearly coupled to a massless quantized radiation field. The corresponding Hamiltonians $H^{\mathrm{PF}}_v$ acts on the Hilbert space $\mathcal{H}_{\mathrm{PF}} := \mathcal{K} \otimes \Gamma (\mathfrak{h} )$, where $\mathcal{K}$ is the Hilbert space for the small quantum system, and $\Gamma( \mathfrak{h} )$ is the symmetric Fock space over $\mathfrak{h} := \mathrm{L}^2( \mathbb{R}^d , {\rm d}k )$. The latter describes a field of massless scalar bosons and is defined by
\begin{equation}\label{eq:def_Fock}
\Gamma( \mathfrak{h} ) := \mathbb{C} \oplus \bigoplus_{n=1}^{+\infty} \otimes^n_s \mathfrak{h},
\end{equation}
where $\otimes^n_s$ denotes the symmetric $n$th tensor product of $\mathfrak{h}$. The operator $H^{\mathrm{PF}}_v$ depends on the form factor $v$ and is written as
\begin{equation}
H^{\mathrm{PF}}_v := K \otimes \mathds{1}_{\Gamma( \mathfrak{h} )} + \mathds{1}_\mathcal{K} \otimes {\rm d} \Gamma( |k| ) + \phi(v),
\end{equation}
where $K$ is a bounded below operator on $\mathcal{K}$ describing the dynamics of the small system, ${\rm d}\Gamma(|k|)$ is the second quantization of the operator of multiplication by $|k|$, and $\phi(v) := ( a^*(v) + a(v) )/ \sqrt{2}$. The form factor $v$ is a linear operator from $\mathcal{K}$ to $\mathcal{K} \otimes\mathfrak{h}$, and $a^*(v)$, $a(v)$ are the usual creation and annihilation operators associated with $v$ (see \cite{BD,GGM}). For convenience, we assume that
\begin{equation}
K \ge 0.
\end{equation}
The hypotheses we make are slightly stronger than the ones considered in \cite{GGM}. The first one, Hypothesis ${\bf (H0)}$, is related to the fact that the small system is assumed to be confined:
\begin{itemize}
\item[{\bf (H0)}] $( K + 1 )^{-1}$ is compact on $\mathcal{K}$.
\end{itemize}
Let $0 \le \tau < 1/2$ be fixed. Let $\mathcal{O}_\tau \subseteq \mathcal{B} ( \D ( K^\tau ) ; \mathcal{K} \otimes \mathfrak{h} )$ be the set of operators which extend by continuity from $\D( K^\tau )$ to an element of $\mathcal{B} ( \mathcal{K} ; \D( K^\tau )^* \otimes \mathfrak{h} )$, that is
\begin{align}
\mathcal{O}_\tau & := \big \{ v \in \mathcal{B} ( \D ( K^\tau ) ; \mathcal{K} \otimes \mathfrak{h} ) , \notag \\
& \exists C > 0 , \forall \psi \in \D( K^\tau ) , \big \| [ ( K + 1 )^{-\tau} \otimes \mathds{1}_{ \mathfrak{h} } ] v  \psi \big \|_{ \mathcal{K} \otimes \mathfrak{h} } \le C \| \psi \|_{ \mathcal{K} } \big \}. \label{eq:def_Otau}
\end{align}
Our first assumption on the form factor is the following:
\begin{itemize}
\item[{\bf (I1)}] $v$ and $[ \mathds{1}_{ \mathcal{K} } \otimes |k|^{-1/2} ] v$ belong to $\mathcal{O}_\tau$.
\end{itemize}
It follows from \cite[Proposition 4.6]{GGM} that, if  $[ \mathds{1}_{ \mathcal{K} } \otimes |k|^{-1/2} ] v \in \mathcal{O}_\tau$, then $H^{\mathrm{PF}}_v$ is self-adjoint with domain
\begin{equation}
\D(H^{\mathrm{PF}}_v) = \D(H^{\mathrm{PF}}_0) = \D( K ) \otimes \Gamma( \mathfrak{h} ) \cap \mathcal{K} \otimes \D( \d \Gamma ( |k| ) ).
\end{equation}

We consider the unitary operator 
\begin{equation}
T : \mathrm{L}^2( \mathbb{R}^d ) \to \mathrm{L}^2( \mathbb{R}^+ ) \otimes \mathrm{L}^2( S^{d-1} ) =: \tilde{ \mathfrak{h} }
\end{equation}
defined by $(Tu)( \omega , \theta ) = \omega^{ (d-1) / 2 } u ( \omega \theta )$. Lifting it to the full Hilbert space $\mathcal{H}_{ \mathrm{PF} }$ by setting $\mathcal{T} := \mathds{1}_{ \mathcal{K} } \otimes \Gamma( T )$ (recall that $\Gamma( T )$ is defined by its restriction to the $n$-bosons Hilbert space as $\Gamma( T ) |_{ \otimes^s_n \mathfrak{h} } = T \otimes \dots \otimes T$ for $n \ge 1$, and $\Gamma( T ) |_{ \mathbb{C} } = \mathds{1}_{ \mathbb{C} }$ for $n=0$), we get a unitary map
\begin{equation}
\mathcal{T} : \mathcal{H}_{ \mathrm{PF} } \to \tilde{ \mathcal{H} }_{ \mathrm{PF} } := \mathcal{K} \otimes \Gamma( \tilde{\mathfrak{h}} ).
\end{equation}
This allows us to write the Hamiltonian in polar coordinates in the following way:
\begin{equation}
\tilde{H}^{\mathrm{PF}}_{ v } :=  \mathcal{T} H^{\mathrm{PF}}_v \mathcal{T}^{-1} = K \otimes \mathds{1}_{\Gamma( \tilde{\mathfrak{h}} )} + \mathds{1}_{\mathcal{K}} \otimes \d \Gamma ( \omega ) + \phi (\tilde v ),
\end{equation}
on $\tilde{\mathcal{H}}_{\mathrm{PF}}$, where 
\begin{equation}
\tilde v := [ \mathds{1}_{ \mathcal{K} } \otimes T ] v
\end{equation}
is a linear operator from $\mathcal K$ to  $\mathcal{K} \otimes  \tilde{\mathfrak{h}}$, and $\d \Gamma( \omega )$ denotes the second quantization of the operator of multiplication by $\omega \in \mathbb{R}^+$.

Let us consider a function $d \in \mathrm{C}^\infty( ( 0,\infty ) )$ satisfying $d'( \omega ) < 0$, $|d'( \omega )| \le C \omega^{-1} d( \omega )$ for some positive constant $C$, $d( \omega ) = 1$ if $\omega \ge 1$, and $\lim_{\omega \rightarrow 0} d( \omega ) = + \infty$ (see Figure \ref{function_d}).

%
%
\begin{figure}[htbp]
\centering
\resizebox{0.7\textwidth}{!}{\includegraphics{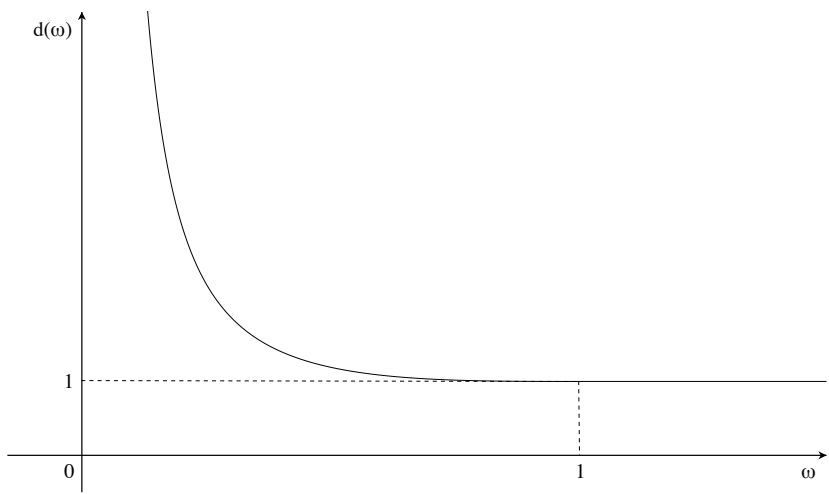}}
\caption{\textbf{The map $\omega \mapsto d(\omega)$}} \label{function_d}
\end{figure}

%
%

Let
\begin{equation} \label{eq:def_tilde_Otau}
\tilde{\mathcal{O}_\tau} := [ \mathds{1}_{ \mathcal{K} } \otimes T ] \mathcal{O}_\tau.
\end{equation}
The following further assumptions on the interaction are made:
\begin{itemize}
\item[{\bf (I2)}] The following holds:
\begin{align*}
& \big [ \mathds{1}_{ \mathcal{K} } \otimes ( 1+\omega^{-1/2} )\omega^{-1} d(\omega) \otimes \mathds{1}_{ \mathrm{L}^2( S^{d-1} ) } \big ]  \tilde v \in \tilde{\mathcal{O}_\tau} \\
& \big [ \mathds{1}_{ \mathcal{K} } \otimes ( 1+\omega^{-1/2} ) d(\omega) \partial_\omega \otimes \mathds{1}_{ \mathrm{L}^2( S^{d-1} ) } \big ] \tilde v \in \tilde{\mathcal{O}_\tau},
\end{align*}
\end{itemize}
\begin{itemize}
\item[{\bf (I3)}] $\big [ \mathds{1}_{ \mathcal{K} } \otimes \partial^2_\omega \otimes \mathds{1}_{ \mathrm{L}^2( S^{d-1} ) } \big ] \tilde v \in \mathcal{B} ( \D ( K^{ \frac{1}{2} } ) ; \mathcal{K} \otimes \tilde{ \mathfrak{h} } )$.
\end{itemize}

Let us recall the definition of the conjugate operator used in \cite{GGM}. Let $\chi \in \mathrm{C}_0^\infty( [0 , \infty ) )$ be such that $\chi( \omega )=0$ if $\omega \ge 1$ and $\chi( \omega )=1$ if $\omega \le 1/2$. For $0<\delta\le 1/2$, the function $m_\delta \in \mathrm{C}^\infty( [0,\infty) )$ is defined by
\begin{equation}
m_\delta( \omega ) = \chi( \frac{ \omega }{ \delta } ) d( \delta ) + (1-\chi)( \frac{ \omega }{ \delta } ) d( \omega ),
\end{equation}
(see Figure \ref{function_m}). 
 
%
%
\begin{figure}[htbp]
\centering
\resizebox{0.7\textwidth}{!}{\includegraphics{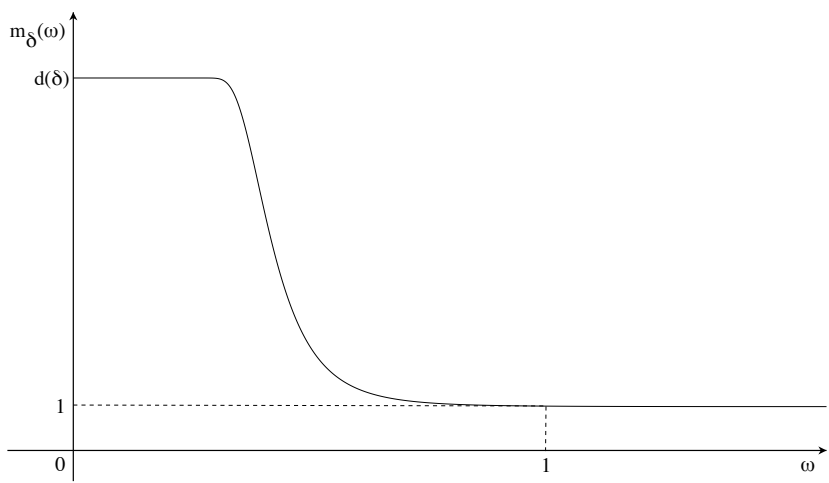}}
\caption{\textbf{The map $\omega \mapsto m_\delta(\omega)$}} \label{function_m}
\end{figure}
%
%

Consider the following operator $\tilde a_\delta$ acting on $\tilde{\mathfrak{h}}$:
\begin{equation}
\tilde a_\delta := \i m_\delta( \omega ) \frac{ \partial }{ \partial \omega } + \frac{ \i }{ 2} \frac{ \d m_\delta}{ \d \omega }( \omega ), \quad \D( \tilde a_\delta ) = \mathrm{H}_0^1( \mathbb{R}^+ ) \otimes \mathrm{L}^2( S^{ d-1}),
\end{equation}
where $H_0^1( \mathbb{R}^+ )$ denotes the closure of $\mathrm{C}_0^\infty( \mathbb{R}^+ )$ in $\mathrm{H}^1( \mathbb{R}^+ )$ and $\mathrm{C}_0^\infty( \mathbb{R}^+ )$ is the set of smooth compactly supported functions on $\mathbb{R}^+$. Then the operator $\tilde A_\delta$ on $\tilde{\mathcal{H}}_{\mathrm{PF}}$ is defined by $\tilde A_\delta := \mathds{1}_{\mathcal{K}} \otimes {\rm d}\Gamma( \tilde a_\delta )$. It is proved in \cite{GGM} that $\tilde A_\delta$ is closed, densely defined and maximal symmetric. 

Let $M_\delta := \mathds{1}_{\mathcal{K}} \otimes {\rm d}\Gamma (m_\delta)$ and $R_\delta( \tilde v ) := - \phi( \i \tilde a_\delta \tilde v )$. Then $M_\delta$ is self-adjoint, $M_\delta \ge 0$, and if $v$ satisfies Hypotheses ${\bf (I1)}$ and ${\bf (I2)}$, then, by \cite[Lemma 6.4 $i)$]{GGM}, $R_\delta( \tilde v )$ is symmetric and $\tilde{H}^{\mathrm{PF}}_{ v }$-bounded.

\subsection{Checking the abstract assumptions}\label{subsection:checking}

In this subsection, we verify that, on the Hilbert space $\mathcal{H} = \tilde{\mathcal{H}}_{\mathrm{PF}}$, the operators $H=\tilde{H}^{\mathrm{PF}}_{ v }$, $M=M_\delta$, $R=R_\delta( \tilde v )$, $A=\tilde A_\delta$ fulfil Conditions \ref{cond:perturbation1}, \ref{cond:perturbation2} and \ref{cond:tech} stated in Section \ref{Introduction} (provided that $v$ satisfies, in particular, the hypotheses stated above). The following lemma shows that Condition \ref{cond:perturbation1} (\ref{item:g1}) is satisfied.
\begin{lemma}\label{lm:HinC1mo}
Assume that $v$ satisfies Hypothesis ${\bf (I1)}$. Then for all $\delta>0$,
\begin{equation}
\tilde{H}^{\mathrm{PF}}_{v} \in \mathrm{C}^1_{\mathrm{Mo}}(M_\delta).
\end{equation}
\end{lemma}
\begin{proof}
The fact that $\tilde{H}^{\mathrm{PF}}_{ v } \in \mathrm{C}^1(M_\delta)$ follows from \cite[Lemma 6.4 $i)$]{GGM}. Moreover, since $m_\delta$ is bounded and $[ \omega , m_\delta ] = 0$, we have that $ [ \tilde{H}^{\mathrm{PF}}_{ v } , \i M_\delta]^0 = - \phi( \i m_\delta \tilde{v} )$ by \cite[Corollary 4.13]{GGM}. Using again that $m_\delta$ is bounded, we then conclude from Hypothesis ${\bf (I1)}$ and \cite[Proposition 4.6]{GGM} that $ [ \tilde{H}^{\mathrm{PF}}_{ v }, \i M_\delta]^0$ is $\tilde{H}^{\mathrm{PF}}_{ 0 }$-bounded, and hence $\tilde{H}^{\mathrm{PF}}_{ v }$-bounded (with relative bound 0).
\end{proof}
Lemma \ref{lm:HinC1mo} together with \cite{GGM} imply:
\begin{prop}\label{prop:pauli1}
Assume Hypothesis ${\bf (H0)}$ and that $v$ satisfies Hypotheses ${\bf (I1)}$, ${\bf (I2)}$ and $\bf{ (I3)}$. Then for all $E_0 \in \mathbb{R}$, there exists $\delta_0>0$ such that for all $0<\delta\le\delta_0$, the operators $H=\tilde{H}^{\mathrm{PF}}_{ v }$, $M=M_\delta$, $R=R_\delta( \tilde v )$, $A=\tilde A_\delta$ fulfil Conditions \ref{cond:perturbation1} with $I = ( -\infty , E_0 )$.
\end{prop}
\begin{remark}
We remark that the formulation of the Mourre estimate stated in \cite[Theorem 7.12]{GGM} is not the same as the one considered in Condition \ref{cond:perturbation1} (\ref{item:g2}). However, one can verify that the latter is indeed a consequence of \cite[Theorem 7.12]{GGM}.
\end{remark}

In order to verify Condition \ref{cond:perturbation2}, we need to impose a further condition on $v$: 
\begin{itemize}
\item[{\bf (I4)}] The form $(K \otimes \mathds{1}_{ \tilde{ \mathfrak{h} } }) \tilde v - \tilde v K$ extends by continuity from $\D ( K \otimes \mathds{1}_{ \tilde{ \mathfrak{h} } } ) \times \D ( K )$ to an element of $\tilde{ \mathcal{O} }_{ \frac{1}{2} }$.
\end{itemize}
Here $\tilde{ \mathcal{O} }_{ \frac{1}{2} }$ is defined as $\tilde{ \mathcal{O} }_{ \tau }$ (see \eqref{eq:def_Otau} and \eqref{eq:def_tilde_Otau}). Notice that, assuming $\bf{ (I1) }$, the statement above is meaningful.

We have to identify the set $\mathcal{B}_{1,\gamma}$ used in Condition \ref{cond:perturbation2}. To this end, let us first introduce some definitions. Let $\mathcal{I}_{\mathrm{PF}}(d)$ be defined by:
\begin{align}
\mathcal{I}_{\mathrm{PF}}(d) := \big \{ v \in \mathcal L ( \mathcal K ; \mathcal K \otimes \mathfrak{h} ), \, & v \text{ satisfies } \mathbf{(I1)}, \mathbf{(I2)}, \mathbf{(I3)}, \mathbf{(I4)} \big \}.
\end{align}
Observe that $\mathcal{I}_{\mathrm{PF}}( d )$ can be equipped with a norm, $\|\cdot\|_{\mathrm{PF}}$, matching the four conditions $\bf{ (I1) }$, $\bf{ (I2) }$, $\bf{ (I3) }$, $\bf{ (I4) }$ (see \cite[Subsection 5.1]{FMS}).

Let $v \in \mathcal{I}_{\mathrm{PF}}(d)$. Let $W_{\delta,t}$ denote the $C_0$-semigroup generated by $\tilde A_\delta$. We set
\begin{align}
\mathcal{G}^{\mathrm{PF}}_\delta := \D( |\tilde{H}^{\mathrm{PF}}_{ v } |^{\frac{1}{2}} ) \cap \D ( M_\delta^{\frac{1}{2}} ).
\end{align}
By Proposition \ref{prop:pauli1}, we have that $H=\tilde{H}^{\mathrm{PF}}_{ v }$, $M=M_\delta$, $A=\tilde A_\delta$ fulfil Condition \ref{cond:perturbation1} (\ref{item:g3}), and hence $W_{\delta,t} |_{\mathcal{G}^{\mathrm{PF}}_\delta}$ is a $C_0$-semigroup. Its generator is denoted by $\tilde A_{\mathcal{G}^{\mathrm{PF}}_\delta}$. Likewise, the extension of $W_{\delta,t}$ to $(\mathcal{G}^{\mathrm{PF}}_\delta)^*$ is a $C_0$-semigroup whose generator is denoted by $\tilde A_{ ( \mathcal{G}^{\mathrm{PF}}_\delta )^* }$.

Let $\mathcal{V}_{1}^{\mathrm{PF}}$ denote the set of symmetric operators $V$, $\epsilon$-bounded relatively to $\tilde{H}^{\mathrm{PF}}_{ v }$, such that $V \in \mathrm{C}^1( \tilde A_{\mathcal{G}^{\mathrm{PF}}_\delta} ; \tilde A_{ ( \mathcal{G}^{\mathrm{PF}}_\delta ) ^* } )$ and $[V , \i \tilde A_\delta]^0$ is $\tilde{H}^{\mathrm{PF}}_{ v }$-bounded. It is equipped with the norm
\begin{equation}
\| V \|_{1}^{\mathrm{PF}} = \| V ( \tilde{H}^{\mathrm{PF}}_{ v } - \i )^{-1} \| + \| [ V , \i \tilde A_\delta ]^0 ( \tilde{H}^{\mathrm{PF}}_{ v } - \i )^{-1} \|.
\end{equation}
By \cite[Proposition 4.6]{GGM}, if $w$ satisfies Hypothesis $\bf{ (I1) }$, then $\phi( \tilde w )$ is $\epsilon$-bounded relatively to $\tilde{H}^{\mathrm{PF}}_{ v }$, and, by \cite[Lemma 6.4 $i)$]{GGM}, if in addition $w$ satisfies Hypothesis $\bf{ (I2) }$, then, for any $\delta > 0$, $[\phi(\tilde w) , \i \tilde A_\delta ]^0 = - \phi ( \i \tilde a_\delta \tilde w )$ is $\tilde{H}^{\mathrm{PF}}_{ v }$-bounded. Moreover, one can verify that the map
\begin{equation}\label{eq:continuous_map}
\mathcal{I}_{\mathrm{PF}}(d) \ni w \mapsto \phi( \tilde w ) \in \mathcal{V}_{1}^{\mathrm{PF}}
\end{equation}
is continuous (see \cite[Lemma 5.8]{FMS}).

In a separate paper, \cite{FMS}, we prove (see \cite[Theorem 5.2]{FMS}):
\begin{prop}\label{prop:pauli2}
Assume Hypothesis ${\bf (H0)}$ and let $v \in \mathcal{I}_{ \mathrm{PF} }(d)$. For all $E_0 \in \mathbb{R}$, there exists $\delta_0>0$ such that for all $0<\delta\le\delta_0$, the operators $H=\tilde{H}^{\mathrm{PF}}_{ v }$, $M=M_\delta$, $R=R_\delta( \tilde v )$, $A=\tilde A_\delta$ fulfil Condition \ref{cond:perturbation2}. Here $I = ( -\infty , E_0 )$ and $\mathcal{B}_{1,\gamma}$ is given by
\begin{align}
B_{1,\gamma} = \{ \phi(\tilde w) , w \in \mathcal{I}_{\mathrm{PF}}(d) , \| w \|_{\mathrm{PF}} \le \tilde \gamma \},
\end{align}
where $\tilde \gamma > 0$ is fixed sufficiently small.
\end{prop}
\begin{remarks}\label{rk:pauli1}
\begin{enumerate}[\quad\normalfont 1)]
\item Since the map \eqref{eq:continuous_map} is continuous, for any $\gamma>0$, the set $\mathcal{B}_{1,\gamma}$ is included in $\{ V \in \mathcal{V}_{1}^{\mathrm{PF}} , \| V \|_{1}^{\mathrm{PF}} \le \gamma \}$ provided that $\tilde \gamma$ is chosen small enough. Moreover, $B_{1,\gamma}$ is clearly star-shaped and symmetric with respect to 0. Hence the requirements of Condition \ref{cond:perturbation2} are satisfied.
\item Under the conditions of Proposition \ref{prop:pauli2}, we do not
  expect Condition \ref{cond:perturbation2c} to be satisfied in
  general. Indeed, the assumption that $v \in \mathcal{I}_{
    \mathrm{PF} }(d)$ in the statement of Proposition
  \ref{prop:pauli2} allows us to control two commutators of
  $\tilde{H}^{ \mathrm{PF}}_v$ with $\tilde A_\delta$. In order to be
  able to conclude that Condition \ref{cond:perturbation2c} is
  satisfied using the method of \cite{FMS}, one would need to control
  three commutators of $\tilde{H}^{ \mathrm{PF}}_v$ with $\tilde A_\delta$ (see \cite{FMS}). This would require a stronger restriction on the infrared behavior of the form factor $v$ than the one imposed by Hypotheses ${\bf (I1)}$--${\bf (I2)}$--${\bf (I3)}$.
\end{enumerate}
\end{remarks}

In order to apply Theorems \ref{thm:upper} and \ref{thm:fermi-golden-rule}, it remains to verify Condition \ref{cond:tech}. Let
\begin{equation}
\mathcal S = \D ( K ) \otimes \Gamma_{\mathrm{fin}} ( \mathrm{C}_0^\infty( \mathbb R^+ ) \otimes \mathrm{L}^2( S^{d-1} ) ),
\end{equation}
where for $\mathcal E \subseteq \mathrm{L}^2( \mathbb{R}^+ ) \otimes \mathrm{L}^2( S^{d-1} )$,
\begin{align}
\Gamma_{\mathrm{fin}} ( \mathcal E ) := \big \{ & \Phi = ( \Phi^{(0)} , \Phi^{(1)} , \Phi^{(2)} , \dots ) \in \Gamma ( \mathcal E ) , \exists n_0 , \Phi^{(n)} = 0 \mfor n \ge n_0 \big \}.
\end{align}
For any $\delta > 0$, $\mathcal S$ is included in $\D(\tilde{H}^{\mathrm{PF}}_{ v }) \cap \D(M_\delta) \cap \D( \tilde A_\delta )$. Moreover, $\mathcal S$ is a core for $\tilde A_\delta^*$. Therefore we get:
\begin{prop}\label{prop:pauli3}
Assume that $v$ satisfies Hypothesis ${\bf (I1)}$. Then, for all $\delta>0$, the operators $H=\tilde{H}^{\mathrm{PF}}_{ v }$, $M=M_\delta$, $A=\tilde A_\delta$ fulfil Condition \nolinebreak \ref{cond:tech}.
\end{prop}

Let us finally mention the particular case for which the unperturbed Hamiltonian under consideration is the non-interacting one, $\tilde{H}^{\mathrm{PF}}_0$, given by
\begin{equation}\label{eq:non-interacting}
\tilde{H}_0^{\mathrm{PF}} := K \otimes \mathds{1}_{ \Gamma ( \tilde{\mathfrak{h}} ) } + \mathds{1}_{ \mathcal{K} } \otimes \d \Gamma( \omega ).
\end{equation}
In this case, one can choose $M=\mathds{1}_{\mathcal{K}} \otimes \mathcal{N}$, where $\mathcal{N} := {\rm d}\Gamma( \mathds{1}_{ \tilde{\mathfrak{h}} } \, )$ is the number operator, and $A=\mathds{1}_{\mathcal{K}} \otimes {\rm d} \Gamma ( \i \partial_\omega)$. Then one can easily check the following proposition:
\begin{prop}\label{prop:unperturbed}
Assume Hypothesis ${\bf (H0)}$. Then the operators $H=\tilde{H}^{\mathrm{PF}}_0,M=\mathds{1}_{\mathcal{K}} \otimes
\mathcal{N},R=0,A=\mathds{1}_{\mathcal{K}} \otimes {\rm d} \Gamma
(\i \partial_\omega)$ fulfil Conditions \ref{cond:perturbation1} (with $I=\mathbb{R}$) and
Condition \ref{cond:perturbation2c}.
\end{prop}
\begin{remark}
The fact that Condition \ref{cond:perturbation2c} is fulfilled under the conditions of Proposition \ref{prop:unperturbed} is obvious, since the unperturbed eigenstates are of the form $\phi \otimes \Omega$, where $\phi$ is an eigenstate of $K$, and $\Omega$ denotes the vacuum in $\Gamma( \tilde{\mathfrak{h}} )$. 
\end{remark}

\subsection{Results}

As a consequence of Propositions \ref{prop:pauli1}, \ref{prop:pauli2} and \ref{prop:pauli3}, applying Theorems \ref{thm:upper} and \ref{thm:fermi-golden-rule}, we obtain:
\begin{thm}\label{thm:fermi-golden-rule-Pauli}
Assume Hypothesis ${\bf (H0)}$. Let $v_0, v \in \mathcal{I}_{ \mathrm{PF} }( d )$. Let $J$ be a compact interval such that $\sigma_{\rm pp}(H^{\mathrm{PF}}_{v_0}) \cap J = \{ \lambda \}$. Let $P_{v_0}$ denote the eigenprojection $P_{v_0}=E_{H^{\mathrm{PF}}_{v_0}}(\{\lambda\})$ and $\bar P_{v_0} = I - P_{v_0}$. Then the following holds:
\begin{itemize}
\item[i)] There exists $\sigma_0>0$ such that for all $0\le |\sigma| \le \sigma_0$, the total multiplicity of the eigenvalues of $H^{\mathrm{PF}}_{v_0} + \sigma \phi( v )$ in $J$ is at most $\dim \mathrm{Ran}(P_{v_0})$.
\item[ii)] Suppose in addition that
  \begin{equation}
    \label{eq:fermi_condition_pauli}
    P_{v_0} \phi( v ) \Im \parb{ (H^{\mathrm{PF}}_{v_0}-\lambda-\i 0)^{-1}\bar P_{v_0}} \phi( v ) P_{v_0} \geq c P_{v_0},
  \end{equation}
for some $c>0$. Then there exists $\sigma_0>0$ such that for all $0 < |\sigma| \le \sigma_0$, 
\begin{equation}
\sigma_{\rm pp}\big ( H^{\mathrm{PF}}_{v_0} + \sigma \phi( v ) \big ) \cap J = \emptyset.
\end{equation}
\end{itemize}
\end{thm}
\begin{remarks}
\begin{enumerate}[\quad\normalfont 1)]
\item In view of Propositions \ref{prop:pauli1}, \ref{prop:pauli2} and \ref{prop:pauli3}, Theorems \ref{thm:upper} and \ref{thm:fermi-golden-rule} imply Theorem \ref{thm:fermi-golden-rule-Pauli} with $\tilde{H}^{\mathrm{PF}}_{ v_0 }$ replacing $H^{\mathrm{PF}}_{v_0}$ and $\phi( \tilde v )$ replacing $\phi( v )$. However, using the unitary transformation mapping $\mathcal{H}_{\mathrm{PF}}$ to $\tilde{\mathcal{H}}_{\mathrm{PF}}$, the statement of Theorem \ref{thm:fermi-golden-rule-Pauli} clearly follows.
\item In the case where the unperturbed Hamiltonian is the non-interac\-ting one, that is $H^{\mathrm{PF}}_{v_0} = H_0^{\mathrm{PF}}$ with $H_0^{\mathrm{PF}} = K \otimes \mathds{1}_{ \Gamma ( \mathfrak{h} ) } + \mathds{1}_{ \mathcal{K} } \otimes \d \Gamma( |k| )$, one can use Proposition \ref{prop:unperturbed} instead of Proposition \ref{prop:pauli2} in order to conclude Theorem \ref{thm:fermi-golden-rule-Pauli} ii). Indeed, it follows from \cite{GGM} that if $v$ satisfies $\bf{ (I1) }$--$\bf{ (I2) }$--$\bf{ (I3) }$, then $\phi(\tilde v) \in \mathcal{V}_2$ (in the sense of Definition \ref{def:W2}). Hence, since Condition \ref{cond:perturbation2c} is satisfied by Proposition \ref{prop:unperturbed}, we can apply Theorem \ref{thm:fermi-golden-rule} with Condition $\ref{item:a})$ instead of Condition $\ref{item:b})$. For a general $v_0 \in \mathcal{I}_{ \mathrm{PF} }(d)$, however, we have to apply Theorem \ref{thm:fermi-golden-rule} with Condition $\ref{item:a})$ (see Remark \ref{rk:pauli1} 2) above).
\end{enumerate}
\end{remarks}
The latter result (the absence of eigenvalues of $H^{\mathrm{PF}}_0 + \sigma \phi( v )$ for sufficiently small $\sigma \neq 0$ according to Fermi Golden Rule) already appears in \cite{DJ} assuming in particular that $\langle \partial_\omega \rangle^\nu \tilde v \in \mathcal B ( \mathcal K ; \mathcal K \otimes \tilde{\mathfrak{h}} )$ for some $\nu > 1$. Let us also mention \cite{Go} where a similar result with different assumptions on the form factor is obtained, still for sufficiently small values of the coupling constant. Besides, in \cite{DJ}, upper semicontinuity of the point spectrum of $H^{\mathrm{PF}}_0 + \sigma \phi( v )$ (in the sense stated in Theorem \ref{thm:fermi-golden-rule-Pauli} i)) is obtained for sufficiently small $\sigma$, assuming that $\langle \partial_\omega \rangle^\nu \tilde v \in \mathcal B ( \mathcal K ; \mathcal K \otimes \tilde{\mathfrak{h}} )$ for some $\nu > 2$. The main achievement of our paper, as far as massless Pauli-Fierz models are concerned, is to provide a method which allows us to consider $H^{\mathrm{PF}}_{v_0}$ as the unperturbed Hamiltonian, for \emph{any} $v_0$ belonging to $\mathcal{I}_{ \mathrm{PF} }(d)$.

A model sharing several properties with the one considered in this subsection is the so-called ``standard model of non-relativistic QED''. For results on spectral theory in this context involving the Mourre method, we refer to \cite{Sk,BFS,BFSS,DJ,FGS}.

\subsection{Example: The massless Nelson model}\label{subsection:Nelson}

An example of a model satisfying the hypotheses of Subsection \ref{subsection:def-Pauli-Fierz} is the Nelson model of confined non-relativistic quantum particles interacting with massless scalar bosons. The Hilbert space is given by
\begin{equation}
\mathcal{H}_{\mathrm{N}} := \mathrm{L}^2( \mathbb{R}^{3P} ) \otimes \mathcal{F},
\end{equation}
where $\mathcal{F} := \Gamma( \mathrm{L}^2( \mathbb{R}^3 ) )$ is the symmetric Fock space over $\mathrm{L}^2( \mathbb{R}^3 )$ (see \eqref{eq:def_Fock}). The Nelson Hamiltonian acts on $\mathcal{H}_{\mathrm{N}}$ and is defined by
\begin{equation}\label{eq:def_HN}
H^{\mathrm{N}}_\rho := K \otimes \mathds{1}_{ \mathcal{F} } + \mathds{1}_{\mathrm{L}^2( \mathbb{R}^{3P} ) } \otimes \d \Gamma( |k| ) + I_\rho( x ). 
\end{equation}
Here $x = (x_1,\dots,x_P)$, and $K$ is a Schr{\"o}dinger operator on $\mathrm{L}^2( \mathbb{R}^{3P} )$ describing the dynamics of $P$ non-relativistic particles. We suppose that $K$ is given by
\begin{equation}
K := \sum_{i=1}^P \frac{1}{2m_i} \Delta_i + \sum_{i<j} V_{ij}( x_i - x_j ) + W(x_1,\dots,x_p),
\end{equation}
where the masses $m_i$ are positive, the confining potential $W$ satisfies
\begin{itemize}
\item[{\bf (W0)}] $W \in \mathrm{L}^2_{\mathrm{loc}}( \mathbb{R}^{3P} )$ and there exist positive constants $c_0,c_1>0$ and $\alpha>2$ such that $W(x) \ge c_0 |x|^{2\alpha} - c_1$,
\end{itemize}
and the pair potentials $V_{ij}$ satisfy
\begin{itemize}
\item[{\bf (V0)}] The $V_{ij}$'s are $\Delta$-bounded with relative bound 0.
\end{itemize}
Without loss of generality, we can assume that $K \ge 0$. Note that ${\bf (W0)}$ implies that Hypothesis ${\bf (H0)}$ of Subsection \ref{subsection:def-Pauli-Fierz} is satisfied.
 
The coupling $I_\rho(x)$ in \eqref{eq:def_HN} is of the form
\begin{equation}
I_\rho(x) := \sum_{i=1}^P \Phi_\rho( x_i ),
\end{equation}
where, for $y \in \mathbb{R}^3$, $\Phi_\rho(y)$ is the field operator defined by
\begin{equation}
\Phi_\rho(y) := \frac{1}{ \sqrt{2} } \int_{ \mathbb{R}^3 } \big ( \rho(k) e^{ -\i k \cdot y } a^*(k) + \bar{ \rho }(k) e^{ \i k \cdot x } a(k) \big ) \d k.
\end{equation}
In particular, $I_\rho(x)$ can be written under the form $I_\rho(x) = \phi( \Psi_{\mathrm{N}}(\rho) )$, where 
\begin{equation*}
\Psi_{\mathrm{N}}(\rho) \in \mathcal{B} ( \mathrm{L}^2( \mathbb{R}^{3P} ) ; \mathrm{L}^2( \mathbb{R}^{3P} ) \otimes \mathrm{L}^2( \mathbb{R}^3 ) ) = \mathcal{B} ( \mathrm{L}^2( \mathbb{R}^{3P} ) ; \mathrm{L}^2( \mathbb{R}^3 ; \mathrm{L}^2 ( \mathbb{R}^{3P} ) ) )
\end{equation*}
is defined by
\begin{equation}\label{eq:def_v-Nelson}
( \Psi_{\mathrm{N}}(\rho) \psi ) (k) ( x_1 , \dots , x_P ) = \sum_{j=1}^P e^{ - \i k \cdot x_j } \rho(k) \psi( x_1 , \dots , x_P ).
\end{equation}
Hence $H^{\mathrm{N}}_\rho$ is a Pauli-Fierz Hamiltonian in the sense of Subsection \ref{subsection:def-Pauli-Fierz}, with $\mathcal{K} = \mathrm{L}^2( \mathbb{R}^{3P} )$ and $\mathfrak{h} = \mathrm{L}^2( \mathbb{R}^3 )$.

For simplicity, we assume that $\rho$ only depends on $k$ through its norm, $|k|$, and, going to polar coordinates, we introduce
\begin{equation}
\tilde \rho( \omega ) = \omega \rho ( \omega , 0 , 0 ), \quad \omega \in \mathbb{R}^+. 
\end{equation}
Our set of conditions on $\tilde \rho$ is the following:
\begin{itemize}
\item[{\bf ($\rho$1)}] $\int_0^\infty ( 1 + \omega^{-1} ) | \tilde{\rho}( \omega ) |^2 \d \omega < \infty$,
\end{itemize}
\begin{itemize}
\item[{\bf ($\rho$2)}] $\int_0^\infty ( 1 + \omega^{-1} ) d( \omega )^2 \big [ \omega^{-2} | \tilde{\rho}( \omega ) |^2 + \big | \frac{ \d \tilde \rho }{ \d \omega } ( \omega ) \big |^2 \big ] \d \omega < \infty$,
\end{itemize}
\begin{itemize}
\item[{\bf ($\rho$3)}] $\int_0^\infty \big | \frac{ \d^2 \tilde{\rho} }{ \d \omega^2 } ( \omega ) \big |^2 \d \omega < \infty$,
\end{itemize}
\begin{itemize}
\item[{\bf ($\rho$4)}] $\int_0^\infty \omega^4 \big | \tilde{\rho} ( \omega ) \big |^2 \d \omega < \infty$,
\end{itemize}
where $d$ denotes the function considered in Subsection \ref{subsection:def-Pauli-Fierz}. Note that  {\bf ($\rho$1)}--{\bf ($\rho$2)}--{\bf ($\rho$3)} are the assumptions made in \cite{GGM}. The further assumption {\bf ($\rho$4)} is made in order that Hypothesis {\bf (I4)} of Subsection \ref{subsection:checking} is satisfied. We observe that {\bf ($\rho$2)} and {\bf ($\rho$4)} imply {\bf ($\rho$1)}.

The set of functions $\rho$ satisfying {\bf ($\rho$1)}--{\bf ($\rho$2)}--{\bf ($\rho$3)}--{\bf ($\rho$4)} is denoted by $\mathcal{I}_{\mathrm{N}}( d )$. The following proposition is proven in \cite[Subsection 5.2]{FMS}:
\begin{prop}\label{prop:Nelson1}
Let $\rho \in \mathcal{I}_{ \mathrm{N} }(d)$. Then $\Psi_{\mathrm{N}}(\rho)$ defined as in \eqref{eq:def_v-Nelson} belongs to $\mathcal{I}_{ \mathrm{PF} }(d)$.
\end{prop}
An example of $\rho$, and hence $\tilde \rho$, satisfying {\bf ($\rho$1)}--{\bf ($\rho$2)}--{\bf ($\rho$3)}--{\bf ($\rho$4)} is
\begin{equation}\label{eq:example_rho}
\rho(k) = e^{ - \frac{ |k|^2 }{ 2 \Lambda^2 } } |k|^{ - \frac{1}{2} + \epsilon }, \quad \tilde \rho( \omega ) = e^{ - \frac{ \omega^2 }{ 2 \Lambda^2 } } \omega^{ \frac{1}{2} + \epsilon },
\end{equation}
with $0 < \Lambda < \infty$ and $\epsilon > 1$.

From Proposition \ref{prop:Nelson1} and Theorem \ref{thm:fermi-golden-rule-Pauli}, we obtain:
\begin{thm}\label{thm:fermi-golden-rule-Nelson}
Assume that Hypotheses ${\bf (W0)}$ and ${\bf (V0)}$ hold. Let $\rho_0 , \rho \in \mathcal{I}_{ \mathrm{N} }(d)$. Let $J$ be a compact interval such that $\sigma_{\rm pp}(H^{\mathrm{N}}_{ \rho_0 }) \cap J = \{ \lambda \}$. Let $P_{ \rho_0 }$ denote the eigenprojection $P_{ \rho_0 }=E_{H^{\mathrm{N}}_{ \rho_0 }}(\{\lambda\})$ and $\bar P_{ \rho_0 } = I - P_{ \rho_0 }$. Then the following holds:
\begin{itemize}
\item[i)] There exists $\sigma_0>0$ such that for all $0\le |\sigma| \le \sigma_0$, the total multiplicity of the eigenvalues of $H^{\mathrm{N}}_{ \rho_0 } + \sigma I_{ \rho }( x )$ in $J$ is at most $\dim \mathrm{Ran}(P_{ \rho_0 })$.
\item[ii)] Suppose in addition that
  \begin{equation}
    \label{eq:fermi_condition_nelson}
    P_{ \rho_0 } I_\rho(x) \Im \parb{ (H^{\mathrm{N}}_{ \rho_0 } - \lambda-\i 0)^{-1}\bar P_{ \rho_0 }} I_\rho(x) P_{ \rho_0 }\geq c P_{ \rho_0 },
  \end{equation}
for some $c>0$. Then there exists $\sigma_0>0$ such that for all $0 < |\sigma| \le \sigma_0$, 
\begin{equation}
\sigma_{\rm pp}\big ( H^{\mathrm{N}}_{ \rho_0 } + \sigma I_\rho(x) \big ) \cap J = \emptyset.
\end{equation}
\end{itemize}
\end{thm}
%
%


In fact, the confinement assumption ${\bf (W0)}$ allows one to make use of a unitary dressing transformation (see e.g. \cite{GGM,FMS}) in order to ``improve'' the infrared behavior of the form factor in the Hamiltonian $H^{\mathrm{N}}_{ \rho_0 }$. More precisely, let {\bf ($\rho$1')} denote the following condition:
\begin{itemize}
\item[{\bf ($\rho$1')}] $\int_0^\infty ( 1 + \omega^{-2} ) | \tilde{\rho}( \omega ) |^2 \d \omega < \infty$.
\end{itemize}
Assuming that $\rho_0$ satisfies this condition, the unitary operator
\begin{equation}
U_{ \rho_0 } := e^{ - \i P \Phi_{ \i \rho_0 / | \cdot | } }
\end{equation}
is well-defined and we can consider the Hamiltonian
\begin{align}
H_{ \rho_0 }^{ \mathrm{N}' } :=& ( \mathds{1}_{ \mathcal{K} } \otimes U_{ \rho_0 } ) H_{ \rho_0 }^{ \mathrm{N} } ( \mathds{1}_{ \mathcal{K} } \otimes U_{ \rho_0 }^* ) \notag \\
=& K_{ \rho_0 } \otimes \mathds{1}_\mathcal{F} + \mathds{1}_{ \mathcal{K} } \otimes \d \Gamma( |k| ) + I_{ \rho_0 }( x ) - I_{ \rho_0 }( 0 ),
\end{align}
where
\begin{align}
K_{ \rho_0 } :=& K + \frac{ P^2 }{ 2 } \int_{ \mathbb{R}^3 } \frac{ | \rho_0(k) |^2 }{ |k| } \d k  - P \sum_{j=1}^P \int_{ \mathbb{R}^3 } \frac{ | \rho_0(k) |^2 }{ |k| } \mathrm{cos}( k \cdot x_j ) \d k.
\end{align}
In the same way as in \eqref{eq:def_v-Nelson}, we observe that $I_{ \rho_0 }(x) - I_{ \rho_0 }( 0 ) = \phi( \Psi'_{\mathrm{N}}( \rho_0 ) )$, where $\Psi'_{\mathrm{N}}( \rho_0 )$ is defined by
\begin{equation}\label{eq:def_v-Nelson-dressing}
( \Psi'_{\mathrm{N}}( \rho_0 ) \psi ) (k) ( x_1 , \dots , x_P ) = \sum_{j=1}^P ( e^{ - \i k \cdot x_j } - 1 ) \rho_0(k) \psi( x_1 , \dots , x_P ).
\end{equation}
In particular, $H_{\rho_0}^{ \mathrm{N}' }$ is a Pauli-Fierz Hamiltonian in the sense of Subsection \ref{subsection:def-Pauli-Fierz}.

We consider the following further conditions:
\begin{itemize}
\item[{\bf ($\rho$2')}] $\int_0^\infty \big | \frac{ \d \tilde{\rho} }{ \d \omega } ( \omega ) \big |^2 \d \omega < \infty$,
\end{itemize}
\begin{itemize}
\item[{\bf ($\rho$3')}] $\int_0^\infty ( 1 + \omega^2 )^{-1} \omega^2 \big | \frac{ \d^2 \tilde{\rho} }{ \d \omega^2 } ( \omega ) \big |^2 \d \omega < \infty$,
\end{itemize}
and we denote by $\mathcal{I}_{\mathrm{N}}'( d )$ the set of functions $\rho$ satisfying {\bf ($\rho$1')}--{\bf ($\rho$2')}--{\bf ($\rho$3')}--{\bf ($\rho$4)}. In \cite[Subsection 5.2]{FMS}, we verify that if $\rho_0 \in \mathcal{I}_{ \mathrm{N}}'(d)$, then $\Psi'_{\mathrm{N}}( \rho_0 )$ defined as in \eqref{eq:def_v-Nelson-dressing} belongs to $\mathcal{I}_{ \mathrm{PF} }(d)$. Notice that for any $0 < \Lambda < \infty$ and $\epsilon > 0$, the function given in \eqref{eq:example_rho} belongs to $\mathcal{I}_{\mathrm{N}}'( d )$.

%
%
%
%
As in the statement of Theorem \ref{thm:fermi-golden-rule-Nelson}, we consider a perturbation of the Hamiltonian $H_{\rho_0}^{\mathrm{N}}$ of the form $\sigma I_\rho(x)$. After the dressing transformation, the perturbation becomes
\begin{align}
\sigma I_{\rho_0,\rho}(x) :=&\sigma ( \mathds{1}_{ \mathcal{K} } \otimes U_{ \rho_0 } ) I_\rho(x) ( \mathds{1}_{ \mathcal{K} } \otimes U_{ \rho_0 }^* ) \notag \\
=& I_\rho(x) - P \sum_{j=1}^P \mathrm{Re} \int_{ \mathbb{R}^3 } \frac{ \bar{ \rho }_0(k) \rho(k) }{ |k| } e^{ - \i k \cdot x_j } \d k.
\end{align}
Notice that $\sigma I_{\rho_0,\rho}(x)$ is \emph{not} a field operator in the sense of Subsection \ref{subsection:def-Pauli-Fierz}. Hence it does not belong to the class of perturbations considered in Theorem \ref{thm:fermi-golden-rule-Pauli}. Nevertheless, proceeding in the same way as what we did in Subsection \ref{subsection:checking} to deduce Theorem \ref{thm:fermi-golden-rule-Pauli} (see in particular \cite[Theorem 1.2 2)]{FMS} for the verification of Condition \ref{cond:perturbation2} in the present context), we obtain:
\begin{thm}\label{thm:fermi-golden-rule-Nelson-dressing}
Assume that Hypotheses ${\bf (W0)}$ and ${\bf (V0)}$ are satisfied and let $\rho_0 \in \mathcal{I}_{ \mathrm{N} }'(d)$ and $\rho \in \mathcal{I}_{ \mathrm{N} }(d)$. Let $J$ be a compact interval such that $\sigma_{\rm pp}(H^{\mathrm{N}}_{ \rho_0 }) \cap J = \{ \lambda \}$. Let $P_{ \rho_0 }$ denote the eigenprojection $P_{ \rho_0 }=E_{H^{\mathrm{N}}_{ \rho_0 }}(\{\lambda\})$ and $\bar P_{ \rho_0 } = I - P_{ \rho_0 }$. Then the conclusions $\mathrm{i)}$ and $\mathrm{ii)}$ of Theorem \ref{thm:fermi-golden-rule-Nelson} hold.
\end{thm}
Observe that, thanks to the unitary dressing transformation $U_{\rho_0}$, the Fermi golden rule condition \eqref{eq:fermi_condition_nelson} is equivalent to the following one:
\begin{equation}
    \label{eq:fermi_condition_nelson-dressing}
    P_{\rho_0}' I_{\rho_0,\rho}(x) \Im \parb{ (H^{\mathrm{N}'}_{ \rho_0 } - \lambda-\i 0)^{-1}\bar P_{\rho_0}'} I_{\rho_0,\rho}(x) P_{\rho_0}' \geq c P_{\rho_0}',
\end{equation}
where 
\begin{equation}
P_{\rho_0}'  := E_{H^{\mathrm{N}'}_{ \rho_0 }}(\{\lambda\}).
\end{equation}
Hence the conclusions of Theorem \ref{thm:fermi-golden-rule-Nelson-dressing} for $H_{\rho_0}^{\mathrm{N}}$ follows from the corresponding statements for $H_{\rho_0}^{ \mathrm{N}' }$.

In Theorem \ref{thm:fermi-golden-rule-Nelson-dressing}, $\rho_0$ and $\rho$ do not belong to the same class of form factors (as far as the infrared singularity is concerned, $\rho_0$ is allowed to have a more singular infrared behavior than $\rho$). This is due to the fact that the unitary transformation $U_{\rho_0}$ is $\rho_0$-dependent, so that the Hamiltonian obtained after the transformation, $H^{ \mathrm{N}' }_{ \rho_0 }$, does not depend linearly on $\rho_0$. Thus, a perturbation of the form $H^{ \mathrm{N}' }_{ \rho_0 + \sigma \rho } - H^{ \mathrm{N}' }_{ \rho_0 }$ does not belong to the class of linear perturbations considered in this paper (at least as far as the Fermi Golden Rule criterion is concerned). Nevertheless, since the non-linear terms in $\sigma$ in the expression of $H^{ \mathrm{N}' }_{ \rho_0 + \sigma \rho } - H^{ \mathrm{N}' }_{ \rho_0 }$ act only on the particle Hilbert space $\mathrm{L}^2( \mathbb{R}^{3P} )$ (and hence, in particular, commute with the conjugate operator $\tilde A_\delta$ of Subsection \ref{subsection:def-Pauli-Fierz}), we expect that the method of this paper can be extended to cover the case where both $\rho_0$ and $\rho$ belong to $\mathcal{I}_{ \mathrm{N}}'(d)$. 

\section{Reduced Limiting Absorption Principle at an
  eigenvalue} \label{Reduced limiting absorption principle at an
  eigenvalue}

In this section we prove two different ``reduced Limiting Absorption Principles''. Assuming Conditions \ref{cond:perturbation1} and \ref{cond:perturbation2b}, we shall prove a Limiting Absorption Principle for the reduced unperturbed Hamiltonian $H \bar P$ (where $P = E_H( \{ \lambda \} )$ and $\bar P = I-P$). If the stronger Condition \ref{cond:perturbation2c} is satisfied, we shall obtain a Limiting Absorption Principle for the reduced perturbed Hamiltonian $H+\alpha P+\sigma V$ (for some $\alpha > 0$), provided that $V\in\mathcal V_2$ and that $\sigma$ is sufficiently small.

Let us first recall from \cite{GGM1}:
\begin{thm}
  \label{thm:reduc-limit-absorpt1}  Assume that Conditions \ref{cond:perturbation1} hold. Suppose  $J \subseteq I$ is   a compact interval
   such that $\sigma_\pp( H)\cap J=
  \emptyset$. Let $S=\{ z \in \mathbb{C} , \Re z\in J,0<|\Im z|\leq 1\}$. For any $1/2<s\le 1$,
  \begin{equation}
    \label{eq:13i}
    \sup_{z\in S}\| \langle A \rangle^{-s} (H-z)^{-1} \langle A \rangle^{-s}\|<\infty.
  \end{equation} Moreover the function $S\ni z\to
  \langle A \rangle^{-s} (H-z)^{-1} \langle A \rangle^{-s} \in \mathcal B (\mathcal H)$ is
  uniformly H\"older continuous of order $s-1/2$.
\end{thm}
\begin{remarks}
\begin{enumerate}[\quad\normalfont 1)]
\item Strictly speaking, the Mourre estimate formulated in Condition \ref{cond:perturbation1} (\ref{item:g2}) together with \cite{GGM1} yield that, for any $\eta \in J$, there is a neighbourhood $I_\eta$ such that, for any compact interval $J_\eta \subseteq I_\eta$, the Limiting Absorption Principle \eqref{eq:13i} holds with $J_\eta$ replacing $J$. The statement of Theorem \ref{thm:reduc-limit-absorpt1} then follows from the compactness of $J$ and a covering argument (see Step II in the proof of Theorem \ref{thm:reduc-limit-absorpt3} below for the use of the same argument).
\item The result \cite[Theorem 3.3]{GGM1} is stronger in that the bound
(\ref{eq:13i}) holds in a stronger
operator topology (given in terms of the Hilbert spaces $\mathcal{ G}$ and
$\mathcal{ G}^*)$. For our purposes (\ref{eq:13i}) suffices. A similar
remark is due for the bounds (\ref{eq:13}) and \eqref{eq:13b} given below. 
\end{enumerate}
\end{remarks}

We shall now obtain a result similar to Theorem \ref{thm:reduc-limit-absorpt1} for a reduced resolvent.
\begin{thm}
  \label{thm:reduc-limit-absorpt2}  Assume that Conditions
  \ref{cond:perturbation1} and Condition \ref{cond:perturbation2b} hold. Suppose $J \subseteq I$ is a compact interval
   such that $\sigma_\pp( H)\cap J=
  \{\lambda\}$. Let $P$ denote the eigenprojection
  $P=E_H(\{\lambda\})$ and let $\bar P=I-P$. Let $S=\{z \in \mathbb{C} , \Re z\in J,0<|\Im z|\leq 1\}$. For any $1/2 < s \le 1$,
  \begin{equation}
    \label{eq:13}
    \sup_{z\in S}\| \langle A \rangle^{-s}(H-z)^{-1}\bar P \langle A \rangle^{-s}\|<\infty.
  \end{equation}
  Moreover there exists $C>0$ such that for all $z,z' \in S$,
    \begin{equation}
    \label{eq:13c_}
    \begin{split}
    & \left \| \langle A \rangle^{-s} \left ( (H-z)^{-1} - (H - z')^{-1} \right ) \bar P \langle A \rangle^{-s} \right \| \le C |z-z'|^{s-\frac{1}{2}}.
    \end{split}
  \end{equation}
\end{thm}
\begin{proof} 
It follows from Conditions \ref{cond:perturbation1} and Condition \ref{cond:perturbation2b} that $\sigma_{ \mathrm{pp} }(H)$ is finite in a neighbourhood of $\lambda$. Hence, possibly by considering a bigger compact interval, we can assume without loss of generality that $\lambda$ is included in the interior of $J$.

Consider Condition \ref{cond:perturbation1} (\ref{item:g2}) with $\eta = \lambda$. Let $J_\lambda \subseteq J$ be a compact neighbourhood of $\lambda$ such that $f_\lambda = 1$ on a neighbourhood of $J_\lambda$. Applying Theorem \ref{thm:reduc-limit-absorpt1} on $[ \inf J , \inf J_\lambda ]$ and using that $P + \bar P = I$, we obtain that
\begin{equation}
\sup_{ z \in \mathbb{C} , \Re z \in [ \inf J , \inf J_\lambda ] , 0 < | \Im z | \le 1 }\| \langle A \rangle^{-s}(H-z)^{-1}\bar P \langle A \rangle^{-s}\|<\infty,
\end{equation}
and that $z \mapsto \langle A \rangle^{-s} (H-z)^{-1} \bar P \langle A \rangle^{-s}$ is H{\"o}lder continuous of order $s-1/2$ on $\{ z \in \mathbb{C} , \Re z \in [ \inf J , \inf J_\lambda ] , 0 < | \Im z | \le 1 \}$. The same holds with $[ \sup J_\lambda , \sup J ]$ replacing $[ \inf J , \inf J_\lambda ]$. Therefore, to conclude the proof, one can verify that it is sufficient to establish the statement of Theorem \ref{thm:reduc-limit-absorpt2} with $J$ replaced by $J_\lambda$. We can follow the proof of \cite[Theorem 3.3]{GGM1}. We emphasize the
differences with \cite{GGM1} and refer the reader to that paper  for more details.

We obtain from  \eqref{eq:Mourre_estimate} with $\eta = \lambda$ that  
\begin{equation}
  \label{eq:Mourrei000}
  M+R\geq 2^{-1} c_0 I  - C_2 f^{\bot}_\lambda(H)^2\inp{H} - f_\lambda(H) K f_\lambda(H).
\end{equation}
Since $f_\lambda(H)$ goes strongly to $P$ as $\lambda \to 0$, we obtain
\begin{equation}
  \label{eq:Mourrei0}
  M+R\geq 3^{-1} c_0 I - C_2f^{\bot}_\lambda(H)^2\inp{H}- C_3P,
\end{equation}
which is valid if the support of $f_\lambda$ is
sufficiently close to $\lambda$.
 Applying $\bar P$ from the left and from the right in
 (\ref{eq:Mourrei0})  yields
\begin{equation}
  \label{eq:Mourrei00}
  \bar P(M+R)\bar P\geq 3^{-1} c_0 \bar P - C_2\bar Pf^{\bot}_\lambda(H)^2\inp{H}\bar P.
\end{equation} 

Next, we can mimic the proof of \cite[Theorem 3.3]{GGM1} using
(\ref{eq:Mourrei00}) and the
following slightly different constructions: In Subsection 3.4 of
\cite{GGM1} the operator $H_\epsilon$ (related to the one from the
seminal paper \cite{Mo}) is taken as $H_\epsilon=H-\i \epsilon H'$.
Notice that here and henceforth we can assume without loss  that
$H'$ is closed (possibly by taking the closure).

We propose to take 
\begin{equation}
  \label{eq:14}
  \bar{H}_\epsilon:=H-\i \epsilon
\bar PH'\bar P,
\end{equation}
with domain $\D(\bar{H}_\epsilon) := \D(H) \cap \D(M) \cap
\mathrm{Ran}(\bar{P})$ on the Hilbert space $\bar{\mathcal{H}} := \bar P \mathcal{H}$. It follows from the assumption
$\mathrm{Ran}(P) \subseteq \D(M)$ that $\bar{H}_\epsilon$ is
well-defined and commutes with $\bar{P}$. Similarly, denoting $H_{\bar P} := H
|_{\D(H) \cap \mathrm{Ran}(\bar{P})}$ and $M_{ \bar P } := (\bar P M \bar P)
|_{\D(M) \cap \mathrm{Ran}(\bar{P})}$, the assumption that $\mathrm{Ran}(P) \subseteq \D(M)$ implies  that $H_{ \bar P }$ and $M_{ \bar P }$ are self-adjoint. Moreover $H_{ \bar P } \in \mathrm{C}^1( M_{ \bar P } )$, $\D( H_{ \bar P } ) \cap \D( M_{ \bar P } )$ is a core for $ M_{ \bar P } $, and $\bar{P} H' \bar{P}$ coincides with the closure of $M_{ \bar P } + R_{ \bar P } $ defined on $\D( M_{ \bar P } ) \cap \D( H_{ \bar P } )$, where $R_{ \bar P } := (\bar P R \bar P)
|_{\D(H) \cap \mathrm{Ran}(\bar{P})}$. Therefore the assumptions of  \cite[Theorem 2.25]{GGM1} are satisfied (see \cite[Lemma 2.26]{GGM1}), which implies that $\bar{H}_\epsilon$ is closed, densely defined, and $\bar{H}_\epsilon^* = \bar{H}_{-\epsilon}$.

Let $\bar{\mathcal{G}} := \mathcal{G} \cap \mathrm{Ran}( \bar P)$. By
Conditions \ref{cond:perturbation1} and the fact that $\mathrm{Ran}(P)
\subseteq \D(M)$, $\bar{H}_\epsilon$ extends to a bounded operator:
$\bar{H}_\epsilon \in \mathcal{B}( \bar{\mathcal{G}} ; \bar{
  \mathcal{G} }^* )$. Mimicking \cite[Subsection 3.4]{GGM1} (replacing
$u \in \D(H_\epsilon)$ in Lemmata 3.9 and 3.10 by $u \in
\D(\bar{H}_\epsilon)$, and using \eqref{eq:Mourrei00}), one can show
that there exists $\epsilon_0$ such that for all $0<|\epsilon| \le
\epsilon_0$, for all $z=\eta+\i\mu$ with $\eta \in J_\lambda$ and $\epsilon \mu > 0$,
$\bar{H}_\epsilon - z$ is invertible with bounded inverse
$\bar{R}_\epsilon(z) \in \mathcal{B}( \bar{\mathcal H} ; \D ( \bar
H_\epsilon ))$. Furthermore $\bar{R}_\epsilon(z)$ extends to a bounded
operator in $\mathcal{B}( \bar{ \mathcal{G} }^* ; \bar{\mathcal{G}} )$
which coincides with the inverse of $( \bar H_\epsilon - z ) \in
\mathcal{B}( \bar{ \mathcal G } ; \bar{ \mathcal{G} }^*)$, and which
satisfies
\begin{subequations}
  \begin{align}
&\| \bar{R}_\epsilon(z) \|_{ \mathcal{B} ( \bar{\mathcal{G}} ; \bar{\mathcal{G}}^* ) } \le \frac{ C }{ | \epsilon | }, \label{eq:||Repsilon||_1} \\
&\| \bar{R}_\epsilon(z) v \|_{ \bar{\mathcal{G}} } \le \frac{ C }{ | \epsilon |^{ \frac{1}{2} } } \left ( \big | ( v , \bar{R}_\epsilon(z) v ) \big |^{ \frac{1}{2} } + \| v \| \right ) \text{ for all } v \in \bar{\mathcal{H}}, \label{eq:||Repsilon||} \\
&\mathrm{s}\text{-} \lim_{\epsilon \to 0^\pm} \bar{R}_\epsilon(z) = ( H_{ \bar P } - z )^{-1} \in \mathcal{B}( \bar{\mathcal{H}} ), \label{eq:s-lim_Repsilon}
\end{align}
\end{subequations}
(see \cite[Proposition 3.11 and Lemma 3.12]{GGM1}).

Let $\rho_\epsilon := \langle \epsilon A \rangle^{s-1} \langle A \rangle^{-s}$, with $1/2<s\le 1$. Instead of looking at the expectation of the resolvent $R_\epsilon
(z):=(H_\epsilon-z)^{-1}$, $\epsilon\neq 0$,  we propose to show a differential inequality
for the quantity
\begin{equation}
  \label{eq:17}
  F_\epsilon(z):=\left \langle \rho_\epsilon u , \bar P \bar R_\epsilon(z) \bar P \rho_\epsilon u \right \rangle;
\end{equation} here $u\in \mathcal{H}$, so that $\rho_\epsilon u \in \mathcal{ D}(A)\subseteq \mathcal{ D}(A^*)$. Note that the assumption $\mathrm{Ran}(P) \subseteq \D(A)$ implies that $\bar P$ leaves $\D(A)$ invariant.

In the same way as in \cite{GGM1}, one can verify that
\begin{align}
\frac{\rm d}{{\rm d} \epsilon} F_\epsilon(z) =& \langle (\frac{\rm d}{{\rm d} \epsilon} \rho_\epsilon) u , \bar R_\epsilon(z) \bar P \rho_\epsilon u \rangle + \langle \rho_\epsilon u , \bar R_\epsilon(z) \bar P (\frac{\rm d}{{\rm d} \epsilon} \rho_\epsilon) u \rangle \notag \\
&+ \langle \bar{R}_\epsilon^*(z) \bar P \rho_\epsilon u , A \rho_\epsilon u \rangle - \langle A \rho_\epsilon u , \bar{R}_\epsilon(z) \bar{P} \rho_\epsilon u \rangle \phantom{ \frac{ d }{ d \epsilon } } \notag \\ 
\phantom{ \frac{ d }{ d } } &+ \epsilon \langle \bar{R}_\epsilon^*(z) \bar P \rho_\epsilon u , \left ( H' P A - A P H' - H'' \right ) \bar{R}_\epsilon(z) \bar P \rho_\epsilon u \rangle,  \label{eq:dF_epsilon}
\end{align}
where ${\rm d} \rho_\epsilon / {\rm d} \epsilon = (s-1) \epsilon |A|^2 \langle \epsilon A \rangle^{s-3} \langle A \rangle^{-s}$. In particular 
\begin{equation}
\| {\rm d} \rho_\epsilon / {\rm d} \epsilon \| \le C | \epsilon |^{s-1} \quad \text{and} \quad \| A \rho_\epsilon \| \le C | \epsilon |^{s-1}. \label{eq:drho_epsilon}
\end{equation}
Next it follows from Conditions \ref{cond:perturbation1} and Condition
\ref{cond:perturbation2b} that
\begin{equation*}
 H' P A - A P H' - H''\in \mathcal{B}( \mathcal{G} ; \mathcal{G}^* ). 
\end{equation*}
  This implies
\begin{align}
\left | \frac{\rm d}{{\rm d} \epsilon} F_\epsilon(z) \right | \le& C_1 |\epsilon|^{s-1} \|u\| \left ( \| \bar R_\epsilon (z)\bar P \rho_\epsilon u \| + \| \bar R^*_\epsilon (z) \bar P \rho_\epsilon u \| \right ) \notag \\
&+C_2 |\epsilon| \| \bar R_\epsilon (z)\bar P \rho_\epsilon u \|_\mathcal{G} \| \bar R^*_\epsilon (z) \bar P \rho_\epsilon u \|_\mathcal{G}.
\end{align}
By \eqref{eq:||Repsilon||}, we obtain
\begin{align}
\left | \frac{\rm d}{{\rm d} \epsilon} F_\epsilon(z) \right | & \le C_3 |\epsilon|^{s-1} \|u\| |\epsilon|^{-\frac{1}{2}} \big ( | F_\epsilon(z) |^{ \frac{1}{2} } + \| \bar P\rho_\epsilon u \| \big ) \notag \\
&\quad +C_4 |\epsilon | \Big ( |\epsilon|^{-\frac{1}{2}} \big ( | F_\epsilon(z) |^{ \frac{1}{2} } + \| \bar P\rho_\epsilon u \| \big )\Big )^2 \notag \\
\phantom{ \frac{ d }{ d } } &\le C_5 | \epsilon |^{ s - \frac{3}{2} } \big ( | F_\epsilon(z) | + \| u \|^2 \big ), \label{eq:dF_epsilon_1}
\end{align}
for $0 < | \epsilon | \le \epsilon_0$. Applying Gronwall's lemma, this yields
\begin{equation}\label{eq:|F_epsilon|}
| F_\epsilon(z) | \le C_6 \| u \|^2,
\end{equation}
which combined with \eqref{eq:s-lim_Repsilon} gives
\begin{equation}
\sup_{ z \in \mathbb{C} , \Re z \in J_\lambda, 0 < | \Im z | \le 1 }\| \langle A \rangle^{-s}(H-z)^{-1}\bar P \langle A \rangle^{-s}\|<\infty.
\end{equation}

In order to prove the H{\"o}lder continuity in $z$, we use that, for $0 < \epsilon_1 < \epsilon_0$,
\begin{align}
F_0(z) - F_0(z') = & - \int_0^{\epsilon_1} \frac{ \d }{ \d \epsilon } ( F_{\epsilon}(z) - F_{\epsilon}(z') ) \d \epsilon  \notag \\
& - \int_{\epsilon_1}^{\epsilon_0} \frac{ \d }{ \d \epsilon } ( F_{\epsilon}(z) - F_{\epsilon}(z') ) \d \epsilon + ( F_{\epsilon_0}(z) - F_{\epsilon_0}(z') ). \label{eq:Holder_1}
\end{align}
It follows from \eqref{eq:dF_epsilon_1} and \eqref{eq:|F_epsilon|} that
\begin{equation}
\Big | \int_0^{\epsilon_1} \frac{ \d }{ \d \epsilon } ( F_{\epsilon}(z) - F_{\epsilon}(z') ) \d \epsilon \Big | \le C_7 \epsilon_1^{ s - \frac{1}{2} } \| u \|^2. \label{eq:Holder_2}
\end{equation}
Moreover, using the first  resolvent equation together with \eqref{eq:dF_epsilon}, \eqref{eq:||Repsilon||_1}, \eqref{eq:||Repsilon||}, \eqref{eq:drho_epsilon} and \eqref{eq:|F_epsilon|}, we obtain
\begin{equation*}
\Big | \frac{ \d }{ \d \epsilon } ( F_\epsilon(z) - F_\epsilon(z') ) \Big | \le C_8 | \epsilon |^{ s - \frac{5}{2} } | z - z' | \| u \|^2,
\end{equation*}
which implies
\begin{equation}
\Big | \int_{\epsilon_1}^{\epsilon_0} \frac{ \d }{ \d \epsilon } ( F_{\epsilon}(z) - F_{\epsilon}(z') ) \d \epsilon \Big | \le C_9 \epsilon_1^{ s - \frac{3}{2} } | z - z' | \| u \|^2. \label{eq:Holder_3}
\end{equation}
Finally, the first  resolvent equation and \eqref{eq:||Repsilon||_1} give
\begin{equation}\label{eq:Holder_4}
\big | ( F_{\epsilon_0}(z) - F_{\epsilon_0}(z') ) \big | \le C(\epsilon_0) | z - z' | \| u \|^2, 
\end{equation}
for some positive constant $C(\epsilon_0)$ depending on $\epsilon_0$. Taking $\epsilon_1 = | z - z' |$, Equation \eqref{eq:13c_} follows from \eqref{eq:Holder_1}--\eqref{eq:Holder_4}.
\end{proof}

We have the following stronger result if Condition \ref{cond:perturbation2c} is further assumed.

\begin{thm}
  \label{thm:reduc-limit-absorpt3} 
  Assume that Conditions \ref{cond:perturbation1} and Condition \ref{cond:perturbation2c} hold. Suppose $J \subseteq I$ is a compact interval such that $\sigma_\pp( H ) \cap J \subseteq \{\lambda\}$. Let $P=E_H(\{\lambda\})$ and $V \in \mathcal V_2$. For $\sigma \in \mathbb{R}$, define $H_\sigma := H + \sigma V$ and $\bar H_\sigma := H_\sigma + \alpha_J P$, where $\alpha_J \in \mathbb{R}$ is fixed such that $\alpha_J > \sup J - \inf J$. Let $S=\{z \in \mathbb{C} , \Re z\in J,0<|\Im z|\leq 1\}$. For all $1/2 < s \le 1$, there exists $\sigma_0>0$ such that for all $|\sigma| \le \sigma_0$,
  \begin{equation}
    \label{eq:13b}
    \sup_{z\in S}\| \langle A \rangle^{-s} (\bar H_\sigma-z)^{-1} \langle A \rangle^{-s}\|<\infty.
  \end{equation} 
  Moreover there exists $C>0$ such that for all $\sigma,\sigma' \in [-\sigma_0,\sigma_0]$, for all $z,z' \in S$,
    \begin{align}
    & \left \| \langle A \rangle^{-s} \left ( (\bar H_\sigma-z)^{-1} - (\bar H_{\sigma'} - z')^{-1} \right ) \langle A \rangle^{-s} \right \| \le C \left ( |\sigma - \sigma'|^{s-\frac{1}{2}} + |z-z'|^{s-\frac{1}{2}} \right ).  \label{eq:13c}
  \end{align} 
\end{thm}
\begin{remarks} \label{rk:reduc-limit-absorpt3}
\begin{enumerate}[\quad\normalfont 1)]
\item \label{item:reduc-limit-absorpt3_0}  In the case $\sigma_{ \mathrm{pp} }( H ) \cap J = \emptyset$, we have $P = 0$ and hence $\bar H_\sigma = H_\sigma$. Of course, Condition \ref{cond:perturbation2c} is not required in this case.
\item \label{item:reduc-limit-absorpt3_1} The assumption that $\alpha_J > \sup J - \inf J$ implies that $H + \alpha_J P$ does not have eigenvalues in $J$.
\item \label{item:reduc-limit-absorpt3_2} Equations \eqref{eq:13b}--\eqref{eq:13c} with $\sigma = \sigma' = 0$ yield that
\begin{equation}
\sup_{z\in S}\| \langle A \rangle^{-s} (H-z)^{-1} \bar P \langle A \rangle^{-s}\|<\infty,
\end{equation}
and that $z \mapsto \langle A \rangle^{-s} (H - z)^{-1} \bar P \langle A \rangle^{-s}$ is H{\"o}lder continuous of order $s-1/2$ on $S$. Hence we recover the Limiting Absorption Principles of Theorems \ref{thm:reduc-limit-absorpt1} and \ref{thm:reduc-limit-absorpt2}.
\end{enumerate}
\end{remarks}
\noindent \textbf{Proof of Theorem \ref{thm:reduc-limit-absorpt3}} \,
Considering the Mourre estimate, Condition \ref{cond:perturbation1} (\ref{item:g2}), for any $\eta \in J$, we denote by $J_\eta \subseteq I$ a compact neighbourhood of $\eta$ such that $f_\eta = 1$ on a neighbourhood of $J_\eta$.

\noindent \textbf{Step 1} \, Let us prove that, for any $\eta \in J$, there exists $\sigma_\eta > 0$ such that for all $| \sigma | \le \sigma_\eta$,
\begin{equation}\label{eq:LAP_Jeta}
\sup_{ z \in \mathbb{C}, \Re z \in J_\eta , 0 < | \Im z | \le 1  }\| \langle A \rangle^{-s} (\bar H_\sigma-z)^{-1} \langle A \rangle^{-s} \| < \infty,
\end{equation}
and that the function $( \sigma , z ) \mapsto \langle A \rangle^{-s} (\bar H_\sigma-z)^{-1} \langle A \rangle^{-s}$ is H{\"o}lder continuous of order $s-1/2$ in $\sigma$ and $z$ on $[ - \sigma_\eta , \sigma_\eta ] \times \{ z \in \mathbb{C}, \Re z \in J_\eta , 0 < | \Im z | \le 1 \}$.

Let $\bar H := H + \alpha_J P$. Condition \ref{cond:perturbation2b} implies that $[P,\i A]^0$ extends to a compact operator. Since $H P = \lambda P$ and $H \bar P = \bar H \bar P$, we have
\begin{align}
 f_\eta^\perp(H)^2 \langle H \rangle = & f_\eta^\perp(\bar H)^2 \langle \bar H \rangle + f_\eta^\perp( \lambda )^2 \langle \lambda \rangle P - f_\eta^\perp( \lambda + \alpha_J )^2 \langle \lambda + \alpha_J \rangle P. \label{eq:f(barH)-f(H)}
\end{align}
Using that the second and third terms in the right-hand-side of \eqref{eq:f(barH)-f(H)} are compact, the Mourre estimate \eqref{eq:Mourre_estimate} yields 
\begin{equation}
  \label{eq:Mourreii}
  M+\parb {R+ \alpha_J [P,\i A]^0}\geq c_0 I -C_0 f^{\bot}_\eta(\bar H)^2\langle \bar H \rangle - K'_0,
\end{equation} 
where $K'_0$ is compact. Since $\eta\notin \sigma_\pp(\bar H)$ (see Remark \ref{rk:reduc-limit-absorpt3} \ref{item:reduc-limit-absorpt3_1})), we can put $K'_0=0$ provided we choose the function $f_\eta$ supported in a sufficiently small interval containing $\eta$. We get
\begin{equation}
  \label{eq:Mourreiii_thmLAP3}
  M+\parb {R+ \alpha_J [P,\i A]^0}\geq 2^{-1} c_0 I -C_1f^{\bot}_\eta(\bar H)^2 \langle \bar H \rangle.
\end{equation}

The estimate (\ref{eq:Mourreiii_thmLAP3}) is stable under
perturbation  from the class $\mathcal V_1$. In particular (and more precisely) there exists $\sigma_\eta>0$ such  that if
$|\sigma| \leq \sigma_\eta$, then
\begin{align}
   &M+\parb {R + \sigma V' + \alpha_J  [P,\i A]^0} \geq 3^{-1} c_0 I -C_2f^{\bot}_\eta(\bar H_\sigma)^2 \langle \bar H_\sigma \rangle. \label{eq:Mourreiv_thmLAP3}
\end{align}
Indeed, since $V \in \mathcal{V}_1$, we have that
\begin{align}
\pm V' \leq C_3 \langle H \rangle + C_4 \leq  C_5 + C_6 f_\eta^\perp( \bar H ) \langle \bar H \rangle f_\eta^\perp( \bar H ),
\end{align}
and
\begin{align}
f_\eta^\perp( \bar H ) \langle \bar H \rangle f_\eta^\perp( \bar H ) & \leq C_7 f_\eta^\perp( \bar H ) \langle \bar H_\sigma \rangle f_\eta^\perp( \bar H ) \notag \\
& \leq C_7 f_\eta^\perp( \bar H_\sigma ) \langle \bar H_\sigma \rangle f_\eta^\perp( \bar H_\sigma ) + C_8 |\sigma|. \label{eq:f(Hbar)-f(H)_2}
\end{align}
The first  inequality in \eqref{eq:f(Hbar)-f(H)_2} follows from elementary
interpolation while 
the second inequality follows, for instance, from the Helffer-Sj{\"o}strand functional calculus.

 We set for shortness $H'_\sigma := H' + \sigma V'$, $H''_\sigma
:= H'' + \sigma V''$, $P' :=  [P,\i A]^0$ and $P'' :=  [P',\i
A]^0$. Remark that Conditions \ref{cond:perturbation1}, Condition
\ref{cond:perturbation2c} and the assumption $V \in \mathcal V_2$
imply that 
\begin{equation*}
 H''_\sigma ,\,P'',\,H''_\sigma + \alpha_J P'' \in \mathcal B( \mathcal G ; \mathcal G^* ).
\end{equation*}

Note that equation \eqref{eq:Mourreiv_thmLAP3} can be written 
\begin{equation}\label{eq:MourreLAP3}
\begin{split}
 H'_\sigma + \alpha_J P' \geq 3^{-1} c_0 I - C_2 f_\eta^\perp( \bar H_\sigma ) ^2 \inp{ \bar H_\sigma}.
\end{split}
\end{equation} We emphasize that the constant $C_2$ is independent of
$z$ and $\sigma$.

To prove \eqref{eq:LAP_Jeta}, we can proceed as in the proof of Theorem \ref{thm:reduc-limit-absorpt2}, using \eqref{eq:MourreLAP3} instead of \eqref{eq:Mourrei00}, and replacing $\bar H_\epsilon$ and $F_\epsilon(z)$ in \eqref{eq:14} and \eqref{eq:17} respectively by
\begin{equation}
\bar H_{\sigma,\epsilon} := \bar H_\sigma  - \i \epsilon ( H'_\sigma + \alpha_J P' ),
\end{equation}
and
\begin{equation}
F_{\sigma,\epsilon}(z) := \langle \rho_\epsilon u , \bar R_{\sigma,\epsilon}(z) \rho_\epsilon u \rangle.
\end{equation}
Here we have set 
\begin{equation}
\bar R_{\sigma,\epsilon}(z) := (\bar H_{\sigma,\epsilon} - z)^{-1}
\end{equation}
and, as before, $\rho_\epsilon = \langle \epsilon A \rangle^{s-1} \langle A \rangle^{-s}$. Notice that, by \cite[Theorem 2.25 and Lemma 2.26]{GGM1}, $\bar H_{\sigma,\epsilon}$ is closed, densely defined and satisfies $\bar H_{\sigma,\epsilon}^* = \bar H_{\sigma,-\epsilon}$. Moreover, following \cite[Subsection 3.4]{GGM1}, one can indeed verify that there exists $\epsilon_0$ such that for all $0 < | \epsilon | \le \epsilon_0$ and $z = \eta' + \i \mu$ with $\eta' \in J_\eta$ and $\epsilon \mu > 0$, $\bar H_{\sigma,\epsilon} - z$ is invertible with bounded inverse $\bar R_{\sigma,\epsilon}(z)$ satisfying properties similar to \eqref{eq:||Repsilon||_1}--\eqref{eq:s-lim_Repsilon}. We can compute:
\begin{align}
 \frac{\rm d}{{\rm d} \epsilon} F_{\sigma,\epsilon}(z) =& \langle (\frac{\rm d}{{\rm d} \epsilon} \rho_\epsilon) u , \bar R_{\sigma,\epsilon}(z) \rho_\epsilon u \rangle + \langle \rho_\epsilon u , \bar R_{\sigma,\epsilon}(z) (\frac{\rm d}{{\rm d} \epsilon} \rho_\epsilon) u \rangle \notag \\
&+ \langle \bar{R}_{\sigma,\epsilon}^*(z) \rho_\epsilon u , A \rho_\epsilon u \rangle - \langle A \rho_\epsilon u , \bar{R}_{\sigma,\epsilon}(z) \rho_\epsilon u \rangle \phantom{ \frac{ d }{ d \epsilon } } \notag \\ 
&- \epsilon \langle \bar{R}_{\sigma,\epsilon}^*(z) \rho_\epsilon u , \left ( H''_\sigma + \alpha_J P'' \right ) \bar{R}_{\sigma,\epsilon}(z) \rho_\epsilon u \rangle. \phantom{ \frac{ d }{ d \epsilon } }
\end{align}

 We obtain as in \eqref{eq:dF_epsilon_1} that
\begin{align}
\left | \frac{\rm d}{{\rm d} \epsilon} F_{\sigma,\epsilon}(z) \right | & \le C_9 | \epsilon |^{ s - \frac{3}{2} } \| u \|^2. \label{eq:dF_sigma,epsilon_1}
\end{align}
Estimate \eqref{eq:13b} (with $J_\lambda$ in place of $J$) and the H{\"o}lder continuity in $z$ then follow as in the proof of Theorem \ref{thm:reduc-limit-absorpt2}.

It remains to prove the H{\"o}lder continuity in $\sigma$. We follow again the proof of Theorem \ref{thm:reduc-limit-absorpt2}. For $0 < \epsilon_1 < \epsilon_0$, we have
\begin{align}
F_{\sigma,0}(z) - F_{\sigma',0}(z) = & - \int_0^{\epsilon_1} \frac{ \d }{ \d \epsilon } ( F_{\sigma,\epsilon}(z) - F_{\sigma',\epsilon}(z) ) \d \epsilon  \notag \\
& - \int_{\epsilon_1}^{\epsilon_0} \frac{ \d }{ \d \epsilon } ( F_{\sigma,\epsilon}(z) - F_{\sigma',\epsilon}(z) ) \d \epsilon + ( F_{\sigma,\epsilon_0}(z) - F_{\sigma',\epsilon_0}(z) ). \label{eq:Holder_sigma_1}
\end{align}
The first term in the right-hand-side of \eqref{eq:Holder_sigma_1} is estimated thanks to \eqref{eq:dF_sigma,epsilon_1}, which gives
\begin{equation}\label{eq:Holder_sigma_2}
\Big | \int_0^{\epsilon_1} \frac{ \d }{ \d \epsilon } ( F_{\sigma,\epsilon}(z) - F_{\sigma',\epsilon}(z) ) \d \epsilon \Big | \le C_{10} \epsilon_1^{ s - \frac{1}{2} } \| u \|^2.
\end{equation}
As for the second and third terms on the right-hand-side of \eqref{eq:Holder_sigma_1}, we use that, by the second resolvent equation,
\begin{equation*}
\bar{R}_{\sigma,\epsilon}(z) - \bar{R}_{\sigma',\epsilon}(z) = - (
\sigma - \sigma' ) \bar{R}_{\sigma,\epsilon}(z) ( V -\i \epsilon V' ) \bar{R}_{\sigma',\epsilon}(z).
\end{equation*}
Since $V$ and $V'$ are $H$-bounded by assumption, this implies in the same way as in the proof of \eqref{eq:Holder_3} and \eqref{eq:Holder_4}  that
\begin{equation}
\Big | \int_{\epsilon_1}^{\epsilon_0} \frac{ \d }{ \d \epsilon } ( F_{\sigma,\epsilon}(z) - F_{\sigma',\epsilon}(z) ) \d \epsilon \Big | \le C_{11} \epsilon_1^{ s - \frac{3}{2} } | \sigma - \sigma' | \| u \|^2 ,\label{eq:Holder_sigma_3}
\end{equation}
and
\begin{equation}\label{eq:Holder_sigma_4}
\big | ( F_{\sigma,\epsilon_0}(z) - F_{\sigma',\epsilon_0}(z) ) \big | \le C(\epsilon_0) | \sigma - \sigma' | \| u \|^2.
\end{equation}
 The H{\"o}lder continuity in $\sigma$ follows from \eqref{eq:Holder_sigma_1}--\eqref{eq:Holder_sigma_4} by choosing $\epsilon_1 = | \sigma - \sigma' |$.
 
\noindent \textbf{Step 2} \, Since $J$ is compact, it follows from Step 1 and a covering argument that there exist $\eta_1, \dots , \eta_l$ (with $l < \infty$) such that $J \subseteq J_{\eta_1} \cup \dots \cup J_{\eta_l}$. Taking $\sigma_0 = \min( \sigma_{\eta_1}, \dots , \sigma_{\eta_l} )$, Equation \eqref{eq:13b} and the H{\"o}lder continuity in $\sigma$ follow. The H{\"o}lder continuity in $z$ is a straightforward consequence of the fact that
\begin{equation}
\sum_{n=1}^l ( a_n )^{s-\frac{1}{2}} \le l^{ \frac{3}{2} - s } \big ( \sum_{n=1}^l a_n \big )^{s - \frac{1}{2}},
\end{equation}
for any sequence of positive numbers $(a_n)_{n=1,\dots,l}$, and $1/2 < s \le 1$.
\qed\\
%
%

\section{Upper semicontinuity of point spectrum} \label{Upper
  semicontinuity of point spectrum}
  
In this section we study upper semicontinuity of the point spectrum of $H$. The main result is Theorem \ref{thm:upper} proven below.
  
Let us begin with stating a consequence of Theorem \ref{thm:reduc-limit-absorpt3}, which shows that if the unperturbed Hamiltonian do not have eigenvalues in a compact interval, the same holds for the perturbed Hamiltonian (provided that the perturbation $V$ belongs to $\mathcal{V}_2$).
\begin{cor}\label{cor:stability}
Assume that Conditions \ref{cond:perturbation1} hold. Let $J \subseteq I$ be a compact interval such that $\sigma_{ \mathrm{pp} }( H ) \cap J = \emptyset$. Let $V \in \mathcal{V}_2$. There exists $\sigma_0 > 0$ such that for any $| \sigma | \le \sigma_0$,
\begin{equation}
\sigma_{ \mathrm{pp} }( H + \sigma V ) \cap J = \emptyset.
\end{equation}
\end{cor}
The statement of Corollary \ref{cor:stability} remains true under the
weaker assumption that $V \in \mathcal{V}_1$, provided that a priori
eigenstates of $H+\sigma V$  belong to $\D( M^{1/2} )$. This is a consequence of the Mourre estimate established in the proof of Theorem \ref{thm:reduc-limit-absorpt3} (see \eqref{eq:Mourreiv_thmLAP3}), together with the virial property that $\langle \psi , ( H' + \sigma V' ) \psi \rangle = 0$ which holds for any eigenstate $\psi$ of $H + \sigma V$ satisfying $\psi \in \D( M^{1/2} )$. Hence we have the following:
\begin{cor}\label{cor:stability2}
Assume that Conditions \ref{cond:perturbation1} hold. Let $J \subseteq I$ be a compact interval such that $\sigma_{ \mathrm{pp} }( H ) \cap J = \emptyset$. Let $V \in \mathcal{V}_1$. There exists $\sigma_0 > 0$ such that for any $| \sigma | \le \sigma_0$, the following holds: Suppose that any eigenstate $\psi$ of $H+\sigma V$ associated to an eigenvalue $\lambda \in J$ satisfies $\psi \in \D( M^{1/2} )$, then
\begin{equation}
\sigma_{ \mathrm{pp} }( H + \sigma V ) \cap J = \emptyset.
\end{equation}
\end{cor}
We now turn to the proof of Theorem \ref{thm:upper}. Here we need Condition \ref{cond:perturbation2} and that $V \in \mathcal{B}_{1,\gamma}$ in addition to Conditions \ref{cond:perturbation1}. \\

\noindent{\bf Proof of Theorem \ref{thm:upper}} \,
Let $\lambda \in I$ and $J \subseteq I$ as in the statement of the theorem. 

\noindent \textbf{Step 1} \, Let us prove that, for any $\eta \in J$, there exist $\beta_\eta > 0$ and $\gamma_\eta > 0$ such that, for $\| V \|_1 \le \gamma_\eta$, the total multiplicity of the eigenvalues of $H + V$ in $(\eta - \beta_\eta , \eta + \beta_\eta)$ is at most $\dim \mathrm{Ker} ( H - \eta )$. 

If $\eta$ is an eigenvalue, we proceed as in \cite[Section 2]{AHS} introducing the (finite rank) eigenprojection, say $P$,  corresponding to this eigenvalue and the auxiliary operator $\bar H=H+\alpha_J P$. Here $\alpha_J > \sup J - \inf J$ as in Theorem \ref{thm:reduc-limit-absorpt3}. Then in the same way as in \eqref{eq:Mourreiv_thmLAP3}, for $\| V \|_1 \le \gamma_\eta$ with $\gamma_\eta > 0$ small enough, we have that
\begin{align}
   &M+\parb {R+ \alpha_J [P , \i A ]^0 + [V , \i A ]^0 } \geq 3^{-1} c_0 I - C_1 f^{\bot}_\eta(\bar H+V)^2 \langle \bar H + V \rangle, \label{eq:Mourreiv}
\end{align} 
where $f_\eta \in \mathrm{C}_0^\infty( \mathbb{R} )$ is such that $0 \le f_\eta \le 1$ and $f_\eta = 1$ in a neighbourhood of $\eta$. Let us in the following agree on the convention that $P=0$ and $\bar {H}=H$ if
$\eta\notin \sigma_\pp( H)$. Then (\ref{eq:Mourreiv}) holds
no matter whether $\eta$ is an eigenvalue or not (provided $\|V\|_1$ is sufficiently small and that the support of $f_\eta$ is chosen sufficiently close to $\eta$). 

Now, it suffices to follow the proof of \cite[Theorem 2.5]{AHS},
combining Condition \ref{cond:perturbation2} and
(\ref{eq:Mourreiv}). More precisely, let $m$ be the multiplicity of
$\eta$ and let us assume that $H+V$ has eigenvalues $(\eta_j)$,
$j=1, \dots, m_1$, of total multiplicity $m_1>m$, located in $(\eta-\beta_\eta,\eta+\beta_\eta) \subseteq I$. Let $(\psi_j)$, $j=1, \dots, m_1$,
be an orthonormal set of eigenvectors, $\psi_j$ being associated with
$\eta_j$.  Consider a linear combination $\psi = \sum_j a_j \psi_j$
such that $\| \psi \|=1$ and $P \psi = 0$. Since $V \in \mathcal{B}_{1,\gamma}$, it follows from Condition
\ref{cond:perturbation2} that $\psi \in \D \cap D(A)$, whence
\eqref{eq:Mourreiv} together with Remark \ref{rk:perturbation1}
\ref{item:h2}) yields
\begin{align}
3^{-1} c_0 & \le \left \langle \psi , ( M + R + \alpha_J [P , \i A]^0 + [ V , \i A ]^0 ) \psi \right \rangle + C_1 \Big \| f_\eta^\perp(\bar H + V) \langle \bar H + V \rangle^{1/2} \psi \Big \|^2 \phantom{\big \|^2} \notag \\
& = \i \left \langle ( \bar H + V - \eta ) \psi , A \psi \right \rangle - \i \left \langle A \psi , ( \bar H + V - \eta ) \psi \right \rangle + C_1 \Big \| f_\eta^\perp(\bar H + V) \langle \bar H + V \rangle^{1/2} \psi \Big \|^2 \notag \\
&\le \beta_\eta \left ( 2 \| A \psi \| + C_2 \beta_\eta \right ). \phantom{\big \|^2}
\end{align}
In the second inequality, we used that 
\begin{equation}
\left \| ( \bar H + V - \eta) \psi \right \| = \| \sum_j a_j ( \eta_j - \eta ) \psi_j \| \le \beta_\eta,
\end{equation}
and hence also that
that $\left \| f_\eta^\perp(\bar H + V) \langle \bar H + V \rangle^{1/2} \psi \right \| \le C_3 \beta_\eta$ by the Spectral Theorem, where the constant $C_3$ depends on $\mathrm{supp} ( f_\eta )$. By
Condition \ref{cond:perturbation2}, we obtain a contradiction provided that $\beta_\eta$ is chosen sufficiently small.

\noindent \textbf{Step 2} \, Let us prove that the total multiplicity of the eigenvalues of $H+V$ in $J$ is at most $\dim \mathrm{Ker} (H-\lambda )$.

It follows from Step 1 that, for any $\eta \in [\inf J , \lambda - \beta_\lambda ] \cup [ \lambda + \beta_\lambda , \sup J ]$, there exist $\beta_\eta > 0$ and $\gamma_\eta > 0$ such that, for $\| V \|_1 \le \gamma_\eta$, $H+V$ does not have eigenvalues in $(\eta - \beta_\eta , \eta + \beta_\eta)$. 
Since $[\inf J , \lambda - \beta_\lambda ] \cup [ \lambda + \beta_\lambda , \sup J ]$ is compact, it follows from a covering argument that there exist $\eta_1,\dots,\eta_l$ such that 
\begin{equation}
[\inf J , \lambda - \beta_\lambda ] \cup [ \lambda + \beta_\lambda , \sup J ] \subset \bigcup_{j=1}^l ( \eta_j - \beta_{\eta_j} , \eta_j + \beta_{\eta_j} ).
\end{equation}
Hence, for $\| V \|_1 \le \min( \gamma_{ \eta_1 } , \dots , \gamma_{ \eta_l } )$, $H+V$ does not have eigenvalues in $[\inf J , \lambda - \beta_\lambda ] \cup [ \lambda + \beta_\lambda , \sup J ]$. Applying Step 1 again with $\eta = \lambda$, this concludes the proof.
\qed\\

The next proposition is a consequence of Theorem \ref{thm:upper}. It will be used in Section \ref{Second order perturbation theory}.
\begin{prop}
  \label{prop:upper-semic-point} Assume that Conditions
  \ref{cond:perturbation1} and  Condition \ref{cond:perturbation2} hold. Suppose $\lambda\in \sigma_\pp( H)
  $ and that $J \subseteq I$ is   a compact interval
   such that $\sigma_\pp( H)\cap J=
   \{\lambda\}$.  Let $P = E_H( \{ \lambda \} )$, $\bar P=I-P$ and
   $P_{V,J}=E_{(H+V)_\pp}(J)$ for any $V\in \mathcal V_1$ (with
   sufficiently small norm). Then for any sequence $V^{(n)} \in \mathcal B_{1,\gamma}$ such that $\|V^{(n)}\|_1 \to 0$, 
   \begin{equation}
     \label{eq:1}
     \|\bar PP_{V^{(n)},J}\|\to 0.
   \end{equation}
One of the following two alternatives $i)$ or $ii)$  holds:
 \begin{enumerate}[\normalfont i)]
  \item \label{item:1s}There exists $0 < \gamma' \le \gamma$  such that if $V\in\mathcal B_{1,\gamma}$ and 
  $0\neq \|V\|_1 \leq \gamma'$, then the  operator $H+V$ does not have
  eigenvalues in $J$.
\item \label{item:2s} There exists a sequence of operators
  $V_n\in \mathcal B_{1,\gamma}$
  with $0\neq \|V_n\|_1\to 0$ and a sequence of normalized
  eigenstates, $(H+V_n-\lambda_n)\psi_n=0$, with eigenvalues
  $\lambda_n\to \lambda$, such that for some 
  $\psi_\infty \in \Ran (P)$ we have $\|\psi_n-\psi_\infty\|\to 0$. 
\end{enumerate}
\end{prop}
  \begin{proof} If \eqref{eq:1} fails there exist an $\epsilon>0$, a
    sequence of elements $V^{(n)}\in \mathcal B_{1,\gamma}$ with $0\neq \| V^{(n)} \|_1\to
  0$, a linear combination of
  eigenstates of  $H+V^{(n)}$, viz. $\psi^{(n)} =
  \sum_{j \leq m(n)} a_j^{(n)} \psi_j^{(n)}$, such that 
\begin{equation}
     \label{eq:1p}
     \|\psi^{(n)}\|\leq 1\mand \|\bar P\psi^{(n)} \|>\epsilon.
   \end{equation} Here  $m(n)\leq \dim \Ran (P)$ specifies the
  dimension of the range of  $P_{V^{(n)},J}$.

Due to Theorem \ref{thm:upper} the corresponding eigenvalues, say
$\lambda^{(n)}_j$,  concentrate at $\lambda$. More precisely
\begin{equation}
  \max_{j \leq m(n)} |\lambda^{(n)}_j-\lambda|\to 0 \mfor n\to\infty.
\end{equation}
In particular we have 
\begin{equation}
  \max_{j \leq m(n)} \| (H - \lambda ) \psi^{(n)}_j \| \to
  0,\mand \max_{j \leq m(n)} \| f_\lambda^\perp(H) \psi^{(n)}_j \| \to 0,
\end{equation}
and therefore also
\begin{equation}\label{eq:3}
  \| (H - \lambda ) \psi^{(n)} \| \to
  0,\mand \| f_\lambda^\perp(H) \psi^{(n)} \| \to 0.
\end{equation} 

Next by the
  Banach-Alaoglu Theorem 
    \cite[Theorem 1 on p. 126]{Yo} we can  assume that there exists the weak limit
  $\psi_\infty:=\w-\lim \psi^{(n)}$ (by passing to a subsequence and
  change notation). From the first identity of
  \eqref{eq:3} we learn that $\psi_\infty\in \Ran(P)$. Consequently
  \begin{equation}\label{eq:4}
    \w-\lim \bar P\psi^{(n)}=\bar P\psi_\infty=0.
  \end{equation}

Now we apply a similar argument as the
  one for proving Theorem \ref{thm:upper} now based on
  (\ref{eq:Mourre_estimate}) rather than (\ref{eq:Mourreiv}): Looking at the
  expectation of both sides of (\ref{eq:Mourre_estimate})  in the states
  $\phi_n:=\bar P\psi^{(n)}$, using Remark \ref{rk:perturbation1} \ref{item:h2}), we obtain
\begin{align}
c_0 \| \phi_n \|^2 \le& 2 \| (H-\lambda) \phi_n \| \| A \phi_n \| + C
\| \inp {H}^{1/2}f_\lambda^\perp(H) \phi_n \|^2 + \inp{\phi_n , K_0 \phi_n }.
\end{align}
   Since $K_0$ is compact we obtain from \eqref{eq:4} that 
  $\inp{\phi_n , K_0 \phi_n } \to 0$. 
  By \eqref{eq:12i}, $\| A \phi_n \|$ is uniformly bounded, and
  therefore we conclude in combination with \eqref{eq:3} that
   $\| \phi_n \|  \to  0$. This contradicts \eqref{eq:1p}.

Let us now prove that either \ref{item:1s}) of \ref{item:2s}) holds. If \ref{item:1s}) fails indeed there exists a sequence of normalized
  eigenstates, $(H+V_n-\lambda_n)\psi_n=0$, with eigenvalues
  $\lambda_n\to \lambda$ and with $V_n \in \mathcal B_{1,\gamma}$, $0\neq \| V_n \|_1\to
  0$. Due to \eqref{eq:1} $\|\bar P \psi_n\|\to 0$. By compactness
  there exists $\psi\in\Ran (P)$ such that 
  along some subsequence  $P \psi_{n_k}\to \psi$. Whence
  \begin{equation}
    \|\psi_{n_k}- \psi\|\leq\|\bar P \psi_{n_k}\|+\| P \psi_{n_k}- \psi\|\to 0 \mfor k\to \infty,
  \end{equation}
  and we conclude \ref{item:2s}). 
    \end{proof}
 There is a different version
      of the second part of Proposition \ref{prop:upper-semic-point} given by first fixing
      $V \in \mathcal B_{1,\gamma}$ (but otherwise given under the same
      conditions). Now we  look at the eigenvalue
      problem in  $I$ of the family of perturbed Hamiltonians
$H_\sigma=H+\sigma V$ with  $\sigma\in \R$ and  $|\sigma|>0$
sufficiently small. In this framework there is a similar dichotomy (it can be shown by applying Proposition \ref{prop:upper-semic-point} under the same conditions, replacing $\mathcal{B}_{1,\gamma}$ by the subset $\{ \sigma V , | \sigma | \le \sigma_0 \} \subseteq \mathcal{B}_{1,\gamma}$). 
\begin{cor}
\label{cor:upper-semic-pointii} Assume that Conditions
\ref{cond:perturbation1} and Condition \ref{cond:perturbation2} hold. Suppose $\lambda\in \sigma_\pp( H)
  $ and that $J \subseteq I$ is a compact interval
   such that $\sigma_\pp( H)\cap J=
  \{\lambda\}$. Let $P = E_H( \{ \lambda \} )$ and let $V \in \mathcal{B}_{1,\gamma}$. One of the following two alternatives i) or ii) holds:
\begin{enumerate}[\normalfont i)]
  \item \label{item:1sf}
For some  sufficiently small $\sigma_0>0$ 
 there are no
eigenvalues of $H_\sigma := H+\sigma V$ in $J$  for all $\sigma \in
]-\sigma_0,\sigma_0[\;\setminus\{0\}$.
\item \label{item:2sf} For some sequence of coupling constants , $0 \neq \sigma_n \to 0$, and some sequence of normalized eigenstates $\psi_n$, $(H+\sigma_nV-\lambda_n)\psi_n=0$ with $\lambda_n \to \lambda$, there exists $\psi_\infty \in \mathrm{Ran}(P)$ such that $\| \psi_n - \psi_\infty \| \to 0$.
  \end{enumerate}
\end{cor}
%
%

 \section{Second order perturbation theory} \label{Second order
      perturbation theory}
      
In this section we shall study second order perturbation theory. Our
main interest is the Fermi Golden Rule, which indeed we shall show is a
consequence of having an expansion to second order of any possible
 existing perturbed
eigenvalue near an unperturbed one. 
This  is done in Subsection
\ref{Second order perturbation  theory -- simple case} under
Conditions \ref{cond:perturbation1} and Condition \ref{cond:perturbation2}, in the case where the unperturbed eigenvalue is simple. In the
degenerate case, this is done in Subsection \ref{Fermi golden rule
  criterion -- general case} assuming Condition
\ref{cond:perturbation2c} rather than Condition  \ref{cond:perturbation2}. We do not obtain  an expansion to second
order of the perturbed eigenvalues assuming Condition \ref{cond:perturbation2} only. Nevertheless we shall show a
similar 
version of the Fermi Golden Rule in this case also (done in Subsection
\ref{Fermi golden rule criterion -- general case}).

\subsection{Second order perturbation  theory -- simple
  case} \label{Second order perturbation  theory -- simple case}

\begin{thm}\label{thm:simple}
Assume that Conditions \ref{cond:perturbation1}, Condition \ref{cond:perturbation2}
   and Condition \ref{cond:tech} hold. Suppose $\lambda\in \sigma_\pp( H)
  $ and that $J \subseteq I$ is a compact interval
   such that $\sigma_\pp( H)\cap J=
  \{\lambda\}$. Let $P = E_H( \{ \lambda \} )$, $\bar P = I - P$. Let $V \in \mathcal{B}_{1,\gamma}$.  Suppose
\begin{equation}
  \dim \Ran (P) = 1,\text{ viz. }P=|\psi\rangle \langle \psi|.
\end{equation}
 For  all $1/2 < s \le 1$ and $\epsilon>0$, there exists $\sigma_0>0$ such that if $|\sigma| \le \sigma_0$ and $\lambda_\sigma \in J$ is an eigenvalue of $H_\sigma$, then 
  \begin{equation}\label{eq:2nd_order_lambda}
  \begin{split}
&\left |  \lambda_\sigma - \lambda - \sigma \langle \psi , V \psi \rangle + \sigma^2 \langle V \psi , ( H - \lambda - \i 0 )^{-1} \bar{P} V \psi \rangle \right | \le \epsilon \sigma^2, 
  \end{split}
  \end{equation}
 and there exists a normalized eigenstate $\psi_\sigma$, $H_\sigma \psi_\sigma = \lambda_\sigma \psi_\sigma$, such that
  \begin{equation}\label{eq:2nd_order_phi}
 \left \|  \psi_\sigma - \psi + \sigma ( H - \lambda - \i 0 )^{-1} \bar P V \psi \right \|_{\D( \langle A \rangle^s )^*} \le \epsilon | \sigma |.
 \end{equation}
  \end{thm}
\begin{remarks}\label{rk:Fermi}
\begin{enumerate}[\quad\normalfont 1)]
\item \label{item:fermi1} It is  a consequence of Conditions
  \ref{cond:perturbation1},  Condition \ref{cond:perturbation2b}, 
  Remark \ref{remarks:domain} and Condition
  \ref{cond:tech} that 
  \begin{equation}
    \label{eq:6}
    \Ran
  (VP)\subseteq \D (A)\mfor \text{all }V \in \mathcal{V}_1.
  \end{equation}
 Notice that we
  can compute the commutator form $[V,\i A]$ on $\parb{ \D
    (M^{1/2})\cap \D (H)\cap \D(A^*)}\times \parb{ \D (M^{1/2})\cap \D
    (H)\cap \D(A)}$ by a formula similar to (\ref{eq:vir}). Whence this
  form is given by $V'$, cf. (\ref{eq:5a}), which by assumption is an
  $H$-bounded operator.  
  In combination with Theorem
  \ref{thm:reduc-limit-absorpt2}  (\ref{eq:6}) implies that indeed  
the operator
\begin{equation}\label{eq:5n}
PV ( H-\lambda - \i0 )^{-1} \bar P V P\in \mathcal
B(\mathcal H).
\end{equation}
\item \label{item:fermi2} Due to Theorem \ref{thm:upper} there is at most one eigenvalue $\lambda_\sigma$ of $H_\sigma$ near $\lambda$, and if it exists it is simple.
\end{enumerate}
\end{remarks}
\begin{cor}
  \label{cor:fermi_simple} Under the conditions of Theorem
  \ref{thm:simple} and the condition
  \begin{equation}
    \label{eq:2}\Im \inp{ V \psi , ( H - \lambda - \i 0 )^{-1} \bar{P} V \psi }>0,
    \end{equation} there exists $\sigma_0>0$ such that for all
  $\sigma \in ]-\sigma_0,\sigma_0[\;\setminus\{0\}$
  \begin{equation}
    \label{eq:21a}
  \sigma_\pp( H_\sigma)\cap J= \emptyset. 
  \end{equation}
\end{cor}
\noindent{\bf Proof of Theorem \ref{thm:simple}} \,
Assume by contradiction that \eqref{eq:2nd_order_lambda} does not hold. Then there exist $\epsilon>0$ and a sequence $\sigma_n \to 0$ such that $H_{\sigma_n}$ has an eigenvalue $\lambda_n$ in $J$ satisfying, for all $n$ and for some $\psi \in \mathrm{Ran}(P)$, $\| \psi \|=1$,
\begin{equation}\label{eq:contradiction}
  \begin{split}
&\left |  \lambda_n - \lambda - \sigma_n \inp { \psi , V \psi } + \sigma_n^2 \inp{ V \psi , ( H - \lambda - \i 0 )^{-1} \bar{P} V \psi } \right | \ge \epsilon \sigma_n^2. 
  \end{split}
  \end{equation}
Since $\dim \mathrm{Ran}(P) =1$, \eqref{eq:contradiction}
actually holds for any $\psi \in \mathrm{Ran}(P)$ such that $\| \psi
\|=1$. Let $\psi_n$ be a normalized eigenstate of $H_n :=
H_{\sigma_n}$ associated to $\lambda_n$,  $H_n \psi_n = \lambda_n
\psi_n$. Arguing as in the proof of Proposition
\ref{prop:upper-semic-point} we can  assume that there exists $\tilde \psi \in \mathrm{Ran}(P)$ such that $\| \psi_n - \tilde \psi \| \to 0$. Henceforth we set $\psi = \tilde \psi$. Let $P_n := E_{H_n}( \{ \lambda_n \} )$. It follows from the fact that $\dim \mathrm{Ran} (P) = 1$ together with Theorem \ref{thm:upper} that $\dim \mathrm{Ran}( P_n ) = 1$. Hence $P_n = |\psi_n\rangle \langle \psi_n|$. The equation
    $(H_n-\lambda_n)P_n = 0$ is equivalent to the following system of equations: 
    \begin{equation}\begin{cases}
      \label{eq:22}
      P\big (\sigma_nV+\lambda-\lambda_n\big )P_n=0,\\
\sigma_n\bar PVP_n+(\lambda-\lambda_n)\bar PP_n+(H-\lambda)\bar PP_n=0.
    \end{cases}\end{equation}

Since $\| \psi_n - \psi \| \to 0$, we have $\| \bar P P_n \| \to 0$ and $\| PP_n \| \to 1$. Hence the first equation of \eqref{eq:22} yields
 \begin{equation}\label{eq:lambda-lambda_n}
 \lambda - \lambda_n = O( | \sigma_n | ).
 \end{equation}
 Now, using the second equation of \eqref{eq:22}, we can write, for any $\phi \in \mathcal{H}$ such that $\| \phi \|=1$, and any $1/2<s\le 1$,
\begin{align}
  \| \bar P P_n \phi \|^2 & =  \left | \langle \bar P P_n \phi, \bar P P_n \phi \rangle \right | \notag \\
 & = \left | \langle \bar P P_n \phi , ( H - \lambda - \i 0
   )^{-1} \parb { \sigma_n \bar PVP_n + (\lambda-\lambda_n + \i0) \bar P P_n } \phi \rangle \right | \notag \\
 & \le C | \sigma_n | \left \| \langle A \rangle^{-s} ( H-\lambda - \i 0 )^{-1}\bar P \langle A \rangle^{-s} \right \| \times \notag \\
 & \qquad \| \langle A \rangle^s \bar P P_n \| \left ( \| \langle A \rangle \bar P P_n \| + \| \langle A \rangle \bar P V P_n \| \right ).
\end{align}
Using Condition \ref{cond:perturbation2} and the assumption that $V \in \mathcal{B}_{1,\gamma}$, one can prove that $\|
\langle A \rangle \bar P P_n \|$ and $\| \langle A \rangle \bar P V P_n \|$ are uniformly bounded in $n$. In addition we claim that for $s<1$, $\| \langle A \rangle^s \bar P P_n \| \to 0$ as $n \to \infty$. To prove this, it suffices to use that $\| \langle A \rangle^s ( \langle A \rangle + \i k )^{-1} \| \to 0$ as $k \to \infty$, together with $\| \langle A \rangle \bar P P_n \|$ being uniformly bounded in $n$ and $\| \bar P P_n \| \to 0$ as $n \to \infty$. Therefore by Theorem \ref{thm:reduc-limit-absorpt2},
\begin{equation}\label{eq:barPP_n}
\| \bar P P_n \|^2 = o( | \sigma_n | ).
\end{equation}
Since $\dim \mathrm{Ran} (P) = \dim \mathrm{Ran}( P_n ) = 1$, Equation \eqref{eq:barPP_n} implies
\begin{equation}\label{eq:barPPsi_n}
\| \bar P \psi_n \|^2 = \| \bar P_n \psi \|^2 = o( | \sigma_n | ),
\end{equation}
and in particular also
\begin{equation}\label{eq:barP_nP}
\| \bar P_n P \|^2 = o( | \sigma_n | ),
\end{equation}
where we have set $\bar P_n = I-P_n$. Taking the expectation of the first equation of \eqref{eq:22} in the state $\psi$ gives
\begin{align}
\lambda - \lambda_n &= - \sigma_n \langle \psi , V \psi \rangle + ( \lambda-\lambda_n)(1- \| P_n \psi \|^2) - \sigma_n \langle \psi , V (P_n - P) \psi \rangle \notag \\
&= - \sigma_n \langle \psi , V \psi \rangle + \sigma_n \langle \psi , V \bar P_n \psi \rangle + o( \sigma_n^2 ), \label{eq:lambda_lambda_n_1}
\end{align}
where we used \eqref{eq:lambda-lambda_n} and \eqref{eq:barPPsi_n} in the second equality. Let us write
\begin{equation}\label{eq:decom_PnPsi}
\bar P_n \psi = P \bar P_n \psi - \bar P P_n \psi.
\end{equation}
Estimate \eqref{eq:barP_nP} yields $\| P \bar P_n \psi \| = o( | \sigma_n | )$. Inserting \eqref{eq:decom_PnPsi} and the second equation of \eqref{eq:22} into \eqref{eq:lambda_lambda_n_1}, we obtain
\begin{align}
\lambda - \lambda_n =& - \sigma_n \langle \psi , V \psi \rangle + \sigma_n^2 \langle V \psi , (H-\lambda - \i 0)^{-1} \bar P V P_n \psi \rangle  \notag \\
&+ \sigma_n ( \lambda - \lambda_n ) \langle V \psi , (H-\lambda-\i 0 )^{-1} \bar P P_n \psi \rangle + o( \sigma_n^2 ).
\end{align}
As above we can use $\lambda - \lambda_n = O( | \sigma_n | )$ together with the fact that $\| \langle A \rangle^s \bar P P_n \| \to 0$ for $s<1$ and Theorem \ref{thm:reduc-limit-absorpt2} to obtain
\begin{equation}
\sigma_n ( \lambda - \lambda_n ) \langle V \psi , (H-\lambda-\i 0 )^{-1} \bar P P_n \psi \rangle = o( \sigma_n^2 ).
\end{equation}
Finally, it follows from Condition \ref{cond:perturbation2} and the assumption $V\in \mathcal{B}_{1,\gamma}$ that $\| \langle A \rangle^s V (P_n-P) \psi \| \to 0$ for $s<1$. This leads to
\begin{equation}
\lambda - \lambda_n = - \sigma_n \langle \psi , V \psi \rangle + \sigma_n^2 \langle V \psi , (H-\lambda - \i 0)^{-1} \bar P V \psi \rangle + o(\sigma_n^2),
\end{equation}
which contradicts \eqref{eq:contradiction}, and hence proves \eqref{eq:2nd_order_lambda}. 

It remains to prove \eqref{eq:2nd_order_phi}. Assume, again by contradiction, that \eqref{eq:2nd_order_phi} does not hold. Then there exist $\epsilon > 0$ and a sequence $\sigma_n \to 0$ such that $H_n = H_{\sigma_n}$ has an eigenvalue $\lambda_n \in J$ associated to a normalized eigenstate $\psi_n$ satisfying, for any $\psi \in \mathrm{Ran}(P)$, $\| \psi \|=1$,
 \begin{equation}\label{eq:contradiction2}
 \left \|  \psi_n - \psi + \sigma_n ( H - \lambda - \i 0 )^{-1} \bar P
   V \psi \right \|_{ (\D( \langle A \rangle^s))^*} \ge \epsilon | \sigma_n |.
 \end{equation}
As above we can assume that there exists $\psi \in \mathrm{Ran}(P)$ such that $\| \psi_n - \psi \| \to 0$. Let $\tilde \psi := e^{i \theta_n } \psi$, where $\theta_n \in \mathbb{R}$ is defined by the equation $ \langle \psi,\psi_n \rangle = e^{i\theta_n} | \langle \psi,\psi_n \rangle|$. Using the second equation of \eqref{eq:22}, we can write
\begin{align}
 \psi_n &= P \psi_n + \bar P \psi_n \notag \\
 &= \langle \psi , \psi_n \rangle \psi - ( \lambda - \lambda_n )
 (H-\lambda - \i 0)^{-1} \bar P \psi_n - \sigma_n (H-\lambda - \i0)^{-1} \bar P V \psi_n \notag \\
 &= \tilde \psi - \sigma_n (H-\lambda - \i0)^{-1} \bar P V \tilde \psi + R_n, 
\end{align}
where
\begin{align}
R_n =& \big ( \| P \psi_n \| - 1 ) \tilde \psi - ( \lambda - \lambda_n ) (H-\lambda - \i0)^{-1} \bar P \psi_n \notag \\
&- \sigma_n (H-\lambda - \i0)^{-1} \bar P V ( \psi_n - \tilde \psi \big ).
\end{align}
By arguments similar to the ones used to prove \eqref{eq:2nd_order_lambda}, one can see that $\| R_n \|_{D( \langle A \rangle^s)^*} = o( | \sigma_n | )$ for any fixed $1/2<s<1$, which contradicts \eqref{eq:contradiction2}, and hence proves \eqref{eq:2nd_order_phi}. 
\qed

\subsection{Fermi Golden Rule criterion -- general
  case} \label{Fermi golden rule criterion -- general case}

We begin this section with a result similar to Theorem \ref{thm:simple} that we shall obtain without requiring an hypothesis of simplicity. Here we need Condition \ref{cond:perturbation2c} rather than Condition \ref{cond:perturbation2}.
\begin{thm}\label{thm:general}
Suppose Conditions \ref{cond:perturbation1}, Condition
  \ref{cond:perturbation2c}  and Condition \ref{cond:tech}. Let $V \in \mathcal V_2$. Suppose $\lambda\in \sigma_\pp( H)
  $ and that $J \subseteq I$ is a compact interval
   such that $\sigma_\pp( H)\cap J=
  \{\lambda\}$. Let $P = E_H( \{ \lambda \} )$, $\bar P = I-P$.
  
There exist $C \ge 0$ and $\sigma_0>0$ such that if $|\sigma| \le \sigma_0$ and $\lambda_\sigma \in J$ is an eigenvalue of $H_\sigma = H + \sigma V$, then there exists $\psi \in \mathrm{Ran}(P)$, $\| \psi \|=1$, such that
  \begin{equation}\label{eq:2nd_order_lambdab}
  \begin{split}
&\big |  \lambda_\sigma - \lambda - \sigma \langle \psi , V \psi \rangle + \sigma^2 \langle V \psi , ( H - \lambda - \i 0 )^{-1} \bar{P} V \psi \rangle \big | \le C |\sigma|^{5/2}.
  \end{split}
  \end{equation}
  \end{thm}
 \begin{remarks}\label{rk:general}  \begin{enumerate}[\quad\normalfont 1)]
\item In the simple case, $P = | \psi \rangle \langle \psi |$, \eqref{eq:2nd_order_lambdab} is stronger than \eqref{eq:2nd_order_lambda}.
  \item We do not have an analogue of \eqref{eq:2nd_order_phi} under
    the conditions of Theorem \ref{thm:general}, even if we assume in
    addition $\dim \mathrm{Ran}(P) = 1$. Similarly, cf. Remark
    \ref{rk:Fermi} \ref{item:fermi2}), we do not have upper semicontinuity of point spectrum at $\lambda$ even if $\dim \mathrm{Ran}(P) = 1$.
  \end{enumerate}
 \end{remarks}

\noindent{\bf Proof of Theorem \ref{thm:general}} \,
We can argue in a way similar to the proofs of Proposition 5.2 and Lemma 5.3 in \cite{AHS}. For $\sigma = 0$, there is nothing to prove. Let $\sigma \neq 0$.

As in the proof of Theorem \ref{thm:reduc-limit-absorpt3},  we set $\bar H = H + \alpha_J P$ with $\alpha_J > \sup J - \inf J$, and $\bar H_\sigma = \bar H + \sigma V$. Assume that $\lambda_\sigma \in \sigma_{\rm pp}(H_\sigma)$ and let $\phi_\sigma$ be such that $(H_\sigma - \lambda_\sigma)\phi_\sigma = 0$, $\| \phi_\sigma \|=1$. Hence
\begin{equation}\label{eq:barHsigma-lambdasigma}
( \bar{H}_\sigma - \lambda_\sigma ) \phi_\sigma = \alpha_J P \phi_\sigma.
\end{equation}
By Theorem \ref{thm:reduc-limit-absorpt3}, $\lambda_\sigma \notin \sigma_{\mathrm{pp}}( \bar{H}_\sigma )$, and hence in particular $P\phi_\sigma \neq 0$. Moreover, it follows from \eqref{eq:barHsigma-lambdasigma} that, for any $\epsilon > 0$,
\begin{equation}
P \phi_\sigma = \alpha_J P \big ( \bar{H}_\sigma - \lambda_\sigma - \i \epsilon \big )^{-1} P \phi_\sigma - \i \epsilon \alpha_J P \big ( \bar{H}_\sigma - \lambda_\sigma - \i \epsilon \big )^{-1} \phi_\sigma.
\end{equation}
Letting $\epsilon \to 0$, since $\lambda_\sigma \notin \sigma_{ \mathrm{pp} }( \bar{H}_\sigma )$, we obtain
\begin{equation}\label{eq:Pphi=Qphi}
P \phi_\sigma = \alpha_J P \big ( \bar{H}_\sigma - \lambda_\sigma - \i 0 \big )^{-1} P \phi_\sigma.
\end{equation}
Note that the right-hand-side of \eqref{eq:Pphi=Qphi} is well-defined by Theorem \ref{thm:reduc-limit-absorpt3} since, by Condition \ref{cond:perturbation2c}, $\mathrm{Ran}(P) \subseteq \D (A)$.

Let $\beta := \alpha_J + \lambda - \lambda_\sigma$. Hence $P \big ( \bar{H} - \lambda_\sigma \big ) P = \beta P$. Using twice the second resolvent equation, one easily verifies that, for any $\epsilon > 0$,
\begin{align}
& P \big ( \bar{H}_\sigma - \lambda_\sigma - \i \epsilon \big )^{-1} P \notag \\
& = ( \beta - \i \epsilon )^{-1} P - ( \beta - \i \epsilon )^{-2} \sigma P V P + ( \beta - \i \epsilon )^{-2} \sigma^2 P V \big ( \bar{H}_\sigma - \lambda_\sigma - \i \epsilon \big )^{-1} V P.
\end{align}
Letting $\epsilon \to 0$ and using Theorem \ref{thm:reduc-limit-absorpt3} with $s=1$, this yields
\begin{align}
& P \big ( \bar{H}_\sigma - \lambda_\sigma - \i 0 \big )^{-1} P \notag \\
& = \beta^{-1} P - \beta^{-2} \sigma P V P + \beta^{-2} \sigma^2 P V \big ( \bar{H} - \lambda_\sigma - \i 0 \big )^{-1} V P + R_1, \label{eq:exprQ}
\end{align}
where $R_1$ is a bounded operator on $\mathrm{Ran}(P)$ satisfying $\| R_1 \| \le C_1 | \sigma |^{5/2}$. Note that the right-hand-side of \eqref{eq:exprQ} is well-defined by Theorem \ref{thm:reduc-limit-absorpt3} and Remark \ref{rk:Fermi} \ref{item:fermi2}).

Now let $\psi := \| P \phi_\sigma \|^{-1} P \phi_\sigma$. Multiplying \eqref{eq:exprQ} by $\alpha_J \beta$ and taking the expectation in $\psi$, we obtain thanks to \eqref{eq:Pphi=Qphi}:
\begin{align}
\lambda - \lambda_\sigma =& - \alpha_J \beta^{-1} \sigma \langle \psi , V \psi \rangle \notag \\
&+ \alpha_J \beta^{-1} \sigma^2 \langle V \psi , \big ( \bar{H} - \lambda_\sigma - \i 0 \big )^{-1} V \psi \rangle + \langle \psi , R_1 \psi \rangle.
\end{align}
Using again Theorem \ref{thm:reduc-limit-absorpt3} with $s=1$, this implies
\begin{align}
\lambda - \lambda_\sigma =& - \alpha_J \beta^{-1} \sigma \langle \psi , V \psi \rangle \notag \\
&+ \alpha_J \beta^{-1} \sigma^2 \langle V \psi , \big ( \bar{H} - \lambda - \i 0 \big )^{-1} V \psi \rangle + \langle \psi , R_2 \psi \rangle, \label{eq:lambda-lambdasigma_12}
\end{align}
where $R_2$ is a bounded operator on $\mathrm{Ran}(P)$ satisfying $\| R_2 \| \le C_2 | \sigma |^{5/2}$. In particular, $| \lambda-\lambda_\sigma | \le C_3 |\sigma| $. We then obtain from \eqref{eq:Pphi=Qphi} and \eqref{eq:exprQ} that
\begin{align}
\frac{ \lambda - \lambda_\sigma }{ \sigma } \psi &= - \alpha_J  \beta^{-1} P V P \psi + \alpha_J  \beta^{-1} \sigma P V \big ( \bar{H} - \lambda - \i 0 \big )^{-1} V P \psi + \sigma^{-1} R_2 \psi \notag \\
&= ( - PVP + \sigma R_3 )\psi, \phantom{ \big )^{-1} }
\end{align}
where $R_3$ is an operator on the finite dimensional space $\mathrm{Ran}(P)$ uniformly bounded in $\sigma$. It follows from the usual perturbation theory (see \cite{Ka}) that $\psi$ can be written as $\psi = \psi_1 + \sigma \psi_2$ where $\psi_1$ is an eigenstate of $-PVP$ and $\psi_2 \in \mathrm{Ran}(P)$. Now, multiplying \eqref{eq:lambda-lambdasigma_12} by $\alpha_J^{-1} \beta$ gives
\begin{align}
( \lambda - \lambda_\sigma ) \alpha_J^{-1} \beta &= - \sigma \inp { \psi , V \psi }
+ \sigma^2 \inp { V \psi , ( \bar H - \lambda - \i 0 )^{-1} V \psi }  +
\alpha_J^{-1} \beta \langle \psi , R_2 \psi \rangle \nonumber \\
&= - \sigma \inp { \psi , V \psi } +  \alpha_J^{-1} \sigma^2 \inp { V \psi , P V \psi } + \sigma^2 \inp { V \psi , ( H - \lambda- \i 0 )^{-1} \bar P V \psi } \nonumber \\
&\quad+ \alpha_J^{-1} \beta \langle \psi , R_2 \psi \rangle. \label{eq:general1}
\end{align}
By \eqref{eq:lambda-lambdasigma_12}, we can write 
\begin{equation}
\lambda - \lambda_\sigma = - \sigma \langle \psi , V \psi \rangle + \langle \psi , R_4 \psi \rangle,
\end{equation}
with $\| R_4 \| \le C_4 \sigma^2$, and hence
\begin{align}
( \lambda - \lambda_\sigma ) \alpha_J^{-1} \beta &= (\lambda - \lambda_\sigma) + \alpha_J^{-1} ( \lambda - \lambda_\sigma )^2 \notag \\
& = (\lambda - \lambda_\sigma) + \alpha_J^{-1} \sigma^2 \langle \psi , V \psi \rangle^2 + O( | \sigma |^3 ). \label{eq:(lambda-lambdasigma)2}
\end{align}
Since $\psi = \psi_1 + \sigma \psi_2$ where $\psi_1$ is an eigenstate of $-PVP$, we have
\begin{equation}\label{eq:general2_1}
\inp {\psi , V \psi}^2 - \| PV \psi \|^2 = O( |\sigma| ).
\end{equation}
Therefore,
\begin{equation}\label{eq:general2}
\alpha_J^{-1} \sigma^2 \inp { V \psi , P V \psi } - \alpha_J^{-1} \sigma^2 \langle \psi , V \psi \rangle^2 = O( | \sigma |^3 ).
\end{equation}
Combining Equations \eqref{eq:general1}, \eqref{eq:(lambda-lambdasigma)2} and \eqref{eq:general2}, the statement of the theorem follows.
\qed \newline

We come now to the proof of Theorem \ref{thm:fermi-golden-rule} on the absence of eigenvalues of  the perturbed Hamiltonian $H_\sigma=H+\sigma V$, generalizing Corollary \ref{cor:fermi_simple}: \\

\noindent{\bf Proof of Theorem \ref{thm:fermi-golden-rule}} \,
  Suppose first that Condition \ref{cond:perturbation2c} holds and that $V \in \mathcal V_2$. By Theorem \ref{thm:general}, there exists $\sigma_0>0$ such that if $\lambda_\sigma$ is an eigenvalue of $H_\sigma$ with $|\sigma| \le \sigma_0$, then \eqref{eq:2nd_order_lambdab} is satisfied. Taking the imaginary part of \eqref{eq:2nd_order_lambdab} contradicts \eqref{eq:18}.
  
Suppose now Condition \ref{cond:perturbation2} and that $V \in \mathcal B_{1,\gamma}$. Assume by contradiction that (\ref{eq:21}) is false. Then the second alternative
    \ref{item:2sf})  of  Corollary
    \ref{cor:upper-semic-pointii} holds. Hence we consider a sequence of
    normalized eigenstates
    $\psi_n\to \psi_\infty\in \Ran(P)$ of  a sequence of Hamiltonians
    $H_n:=H_{\sigma_n}$ given in terms of  a certain sequence of
    coupling constants $\sigma_n\to
    0,\;\sigma_n\neq 0$. Let
    $P_n=|\psi_n\rangle \langle \psi_n|$. 
As in the proof of Theorem \ref{thm:simple}, the equation $(H_n-\lambda_n)P_n=0$ is equivalent to \eqref{eq:22}. We notice that 
\begin{align}
  &\Im \parb { P_n V\bar PP_n}=-\Im \parb { P_n V PP_n}=\tfrac{\lambda-\lambda_n}{\sigma_n}\Im \parb { P_n PP_n}=0,\label{eq:23}
\end{align} due to the first equation of (\ref{eq:22}).
Next we apply $P_nV(H-\lambda-\i 0)^{-1}\bar P$ from  the left
in 
the second equation of (\ref{eq:22}), take the imaginary part and use (\ref{eq:23}) yielding
\begin{align}
  &\sigma_nP_n V\Im \parb { (H-\lambda-\i 0)^{-1}\bar P}VP_n\nonumber\\&=(\lambda_n-\lambda)\Im \parb { P_n V (H-\lambda-\i 0)^{-1}\bar PP_n}.\label{eq:23i}
\end{align}

Now we take the expectation of  (\ref{eq:23i}) in  the state $\psi_\infty$,
use the first equation of 
(\ref{eq:22})  and divide by $\sigma_n$ yielding 
\begin{align}
&  \Im  \inp{ (H-\lambda-\i 0)^{-1}\bar
    P}_{VP_n\psi_\infty} \notag \\
 &=\Im 
  \inp{P_nV(H-\lambda-\i 0)^{-1}\bar P
    P_nV}_{\psi_\infty}.
\end{align}
Again, using Condition \ref{cond:perturbation2}, we have that $\| \langle A \rangle^s \bar P
    P_n\|\to 0$ for $1/2<s<1$. We then conclude by letting $n\to \infty$ in the above
    identity, using Theorem \ref{thm:reduc-limit-absorpt2}, which yields
\begin{equation}\label{eq:25}
  \Im  \inp{ (H-\lambda-\i 0)^{-1}\bar
    P}_{V\psi_\infty}=0.
\end{equation} 
Clearly  (\ref{eq:25}) contradicts  (\ref{eq:18}).
\qed


\end{document}